\documentclass[12pt,reqno,a4paper]{amsart}
\usepackage{amsmath,amssymb,amsfonts,amscd}
\usepackage[mathscr]{eucal}
\usepackage{hyperref}
\usepackage[dvipsnames]{xcolor} %% for PDF
\usepackage{pgf,tikz,pgfplots}
\pgfplotsset{compat=1.14}
\usepackage{mathrsfs}
\usepackage{subcaption} %  for subfigures environments
\usetikzlibrary{arrows}
\usepackage[margin=1in]{geometry}

\newtheorem{theorem}{Theorem}

\theoremstyle{definition}
\newtheorem{definition}{Definition}%[section]
\newtheorem{example}{Example}%[section]
\newcommand{\N}{\mathbb{N}}

\newcommand{\p}{{\partial}}

\newcommand{\vecc}[1]{\mathbf{#1}}

\newcommand{\cL}{\mathcal{L}}
\newcommand{\cC}{\mathcal{C}}
\newcommand{\cN}{\mathcal{N}}
\title[Laplace invariants]{Laplace invariants of differential operators
}
\author{D.~Hobby}
\address{Department of Mathematics, SUNY New Paltz, New Paltz, NY, USA}
\email{hobbyd@newpaltz.edu}
\author{E.~Shemyakova}
\address{Department of Mathematics and Statistics, University of Toledo, Toledo,  OH, USA}
\email{ekaterina.shemyakova@utoledo.edu}
\begin{document}
\begin{abstract} 
We identify conditions giving 
large natural classes of partial differential operators for which it is possible to construct a complete set of Laplace invariants. In order to do that we investigate general properties of differential invariants of partial differential operators under gauge transformations and  introduce a sufficient condition for a set of invariants to be complete. We also give a some mild conditions that guarantee the existence of such a set. The proof is constructive. The method gives many examples of invariants  previously known in the literature as well as many new examples including multidimensional. 
\end{abstract}

\maketitle
\section{Introduction}

Gauge transformations ($\varphi \mapsto e^{g} \varphi$) of differential operators are important transformations that
preserve algebraic structure of an operator, such as 
its ``factorizability'' %of an operator 
into factors of some fixed form %and preserve 
or 
existence of Darboux transformations. Invariant properties like these are best described by gauge invariants (algebraic expressions in the coefficients of the operator and their derivatives). 

%To describe gauge transformations ($\varphi \mapsto e^{g} \varphi$) of differential operators one can use gauge differential invariants (algebraic expressions in the coefficients of the operator and their derivations).

The first examples of gauge invariants for differential operators are the Laplace invariants of the hyperbolic second order operator $L = \p_{xy} + a \p_x + b \p_y + c$~\cite{Darboux2}. These are the gauge invariants $h$ and $k$ that can be thought of as derived from \emph{incomplete factorizations} of this operator  as $L = (\p_x + b)(\p_y + a) + h$ and $L = (\p_y + a)(\p_x + b) + k$. These invariants uniquely define the gauge class of the operator, and so are a \emph{complete set of invariants}. %\cat{These invariants are closely connected with factorizations: operator $L$ has a factorization whenever $h=0$ or $k=0$.}

This is the starting point of a method for solution of $Lu=0$ for the above operator in the closed form. Two Laplace transformations $L \mapsto L_1$ and $L \mapsto L_{-1}$ are defined by intertwining relations  $N_1L=L_1(\p_x + b)$ and $N_{-1}L=L_{-1}(\p_y + a)$. Each of the transformations swaps the values of $h$ and $k$ and then changes one of them. In addition, the two Laplace transformations are (up to the gauge equivalence class) inverses of each other. So as the result of consecutive application of Laplace transformations to some operator $L$ we have a chain of the corresponding pairs of invariants (not a lattice as may be expected). This ``Laplace  chain" is finite if one of the invariants is zero at some point in each direction of the chain. This corresponds to factorizability of the transformed operator. The original equation $Lu=0$ then can be solved in closed form invoking the invertibility of Laplace transformations. 
 
Laplace transformations are members of a larger group of transformations --- Darboux transformations --- which can be defined   algebraically by the means of an intertwining relation $NL=L_1M$. For operators $L = \p_{xy} + a \p_x + b \p_y + c$ it was proved~\cite{shem:darboux2,shemyakova2013_DT_fact} (and then a discrete and a semi-discrete analogues of this result was proved
by S.~Smirnov~\cite{Smirnov2018}) that Laplace transformations are the only invertible Darboux transformations and all the others, even corresponding to a higher order operator $M$, are not.
These non-invertible Darboux transformations induce a map of  kernels $\ker L \rightarrow \ker L_1$ which is not monomorphic, so some solutions are lost. A new construction of invertible Darboux transformations for a large class of operators was discovered in~\cite{2013:invertible:darboux} and for even larger class in~\cite{shemya:hobby2016:iterated}. Another method using pseudodifferential operators was proposed in~\cite{ganzha2013intertwining}.
Multidimensional Darboux transformations are proposed by G.~Hovhannisyan et al.~\cite{HOVHANNISYAN20161690,HOVHANNISYAN2018776}. 
Complete classification of Darboux transformations on the superline (operators of arbitrary order and Darboux transformations of arbitrary order) was obtained in \cite{2015:super,shemya:voronov2016:berezinians}. Note that with every manifold one can naturally  associate a commutative algebra consisting of formal sums of densities of arbitrary real weights. It is useful for geometric analysis of differential operators. In~\cite{ShemyakovaVoronov_Sturm_Liou_densities2018}, we studied factorization of differential operators on such algebra in the case of the line, with an eye at extending Darboux transformations theory to them.
%These methods does not have any known connections with gauge invariants. 

%Laplace invariant is a 
Gauge invariants can be found using \emph{regularized moving frames} method of M.~Fels and P.~Olver~\cite{FO1,FO2}, see also E.~Mansfield's book~\cite{Mansf_book}, which was developed later also for pseudo-groups by P.~Olver and J.~Pohjanpelto~\cite{OP:09}. They also proved that the algebra of invariants can be generated by a finite number of invariants and a finite number of invariant derivatives  (which are particular invariant differential operators on the algebra of invariants). % for un-constrained actions of pseudogroups.
%and B.~Kruglikov
 %  and V.~Lychagin's 2016
  % \cite{kruglikov2016LieTresse}
   %proved that algebra of differential invariants for an algebraic action of a pseudogroup can be generated by a finite number of differential invariants and finite number of \emp{invariant derivations}. 
   %For the case of gauge transformations, the invariant derivations can be taken the original derivations on the differential field, i.e. $\p_1, \dots ,\p_n$ in our notation.
The specifics of the use of the regularized moving frames method for gauge invariants of differential operators is described by the second author with E.~Mansfield in~\cite{movingframes}.

Concerning Laplace invariants for differential operators the following results are known.
Dzhokhadze's 2004~\cite{dzhokhadze2004invariants} and Mironov's 2009~\cite{mironov2009invariants} for
    4th order operators and  Ch.~Athorne and H.~Yilmaz's 2016~\cite{AthorneYilmaz2016} for arbitrary order operators of the  %following
    form   $\sum_{|\vecc{v}|=0}^d
   \left(
        \sum_{\forall i,j, v_i \neq v_j } {a_\vecc{v}} \p^{\vecc{v}} 
        \right)$,
where $d$ is the order of the operator. Thus the order of the operator cannot be larger than the number of the independent variables available. For example, in bivariate case the highest possible order is two and such operators have form $\p_{x} \p_y + a_1 \p_x + a_2 \p_y +a_3$;
        for dimension three the highest possible order is three and such operators have form $\p_{x} \p_y \p_z + a_1 \p_x \p_y + a_2 \p_y \p_z + a_3 \p_x \p_z + a_4 \p_x + a_5 \p_y + a_6 \p_z + a_7$. 
        %The coefficients are not necessarily constants.
        Ch.~Athorne and H.~Yilmaz's 2016~\cite{AthorneYilmaz2016} found {some} Laplace invariants for such operators of arbitrary order and of arbitrary dimension.
        Afterwards they constructed and investigated the corresponding Darboux (Laplace) transformations~\cite{Athorne2018,AthorneYilmaz2019}. 
        
 In 2007, the second author with F.~Winkler~\cite{obstacle2} proposed an algebraic structure, a \emph{ring of obstacles}, where 
   the remainders of incomplete factorizations for operators of arbitrary order and arbitrary number of variables become invariants. The method gave in particular Laplace invariants for bivariate operators with principal symbols $ (p\p_x + q \p_y)\p_x \p_y$, $\p_x^2 \p_y$, and $\p_x^3$. The  Laplace invariants set given by this method is not complete; however, we managed to find an extra  (``non-Laplace'') invariant for each case making the resulting sets complete~\cite{invariants_gen,inv_cond_hyper_case,inv_cond_non_hyper_case}. 
   %in the sense that their values could not uniquely identify the gauge class for an operator.

   M.~van Hoijer with students and collaborators, see e.g.,~ \cite{VanHoeij_transformations2017},
   works on the solution methods for linear homogeneous ordinary differential equations with rational function or polynomial coefficients. Such is for example,   the problem of hypergeometric solutions. In~ \cite{VanHoeij_transformations2017}
   and other works, the authors use gauge transformations (they are called there \emph{exponential transformations}) and construct  Darboux transformations and the corresponding Laplace invariants and use them to simplify the equations.
  
Note that there is difference between finding a ring of invariants as  specified by some arbitrary choice of a generating set  and finding   a ``distinguished" generating set whose elements can carry extra information. (The reader can have in mind classical examples of distinguished invariants in differential geometry such as curvature or torsion, or e.g. particular characteristic classes such as   Chern classes, etc.) 
In the literature,    Laplace invariants  typically mean gauge invariants distinguished in this way. Unlike their classical prototype, they   cannot be obtained by a direct generalization of the Laplace method as ``remainders'' of incomplete factorizations: it is known~\cite{obstacle2} that  such
``remainders'' %of incomplete factorizations 
are not invariants for a general operator. (If this remainder is an operator, then even its principal symbol is not invariant in the general situation.) Nevertheless, Laplace type invariants are distinguished in the sense
that they control representability of an operator in some ``generalized'' factorized form as we show here. 
   
In the present paper, the main results are contained in  Theorem~\ref{sufficient_for_a_complete_set_theorem} and Theorem~\ref{main_theorem} which together show 
that under certain rather general and natural conditions an operator has a complete set of Laplace invariants.

The proofs are constructive and provide a general method of constructing complete sets of Laplace invariants for a very large class of operators which include previously considered classes. We show that examples of Laplace invariants existing in the literature can   be obtained by our method, and we also have examples with new types of operators.

The paper is organized as follows.
After preliminaries in Sec.~\ref{sec:preliminaries}, in Sec.~\ref{sec:complete} we define \emph{maximally generated} and \emph{approximately flat} classes of operators, which impose some natural restrictions on operators (can be multidimensional and of arbitrary high order). From the perspective 
of our method, known Laplace invariants can be classified into
four types (classification may be incomplete but we do not need that  here), we call them \emph{maximal}, \emph{extra}, \emph{compatibility}, and 
\emph{upward} invariants. We illustrate them with examples. For these classes of operators we prove Theorem~\ref{sufficient_for_a_complete_set_theorem} that if one has enough number of invariants of each type, then the set of invariants
is complete. In Sec.~\ref{sec:method}
we introduce the method, first by demonstrating it on known and new examples. Informally, a complete set of gauge invariants obtained by any classical method consists of ``nice/short formula'' low degree invariants and of some ``huge formula'' ones of high degrees. The proposed method replaces ``huge'' ones with other ``huge'' that are now associated with some generalized ``incomplete factorization'' of an operator, so now they have structure and meaning. These are what are known as \emph{Laplace invariants} in the literature. In Sec.~\ref{sec:proof} we give a  theoretical justification of the method and prove the main result, Theorem~\ref{main_theorem},  that 
under some natural conditions 
(\emph{maximally generated}, \emph{framed}, \emph{approximately flat}) an operator has a complete set of (Laplace) invariants. The proof is constructive and is essentially the method that is illustrated by  examples in Sec.~\ref{sec:complete}.
%%%%%%%%%%%%%%%%%%%%%%%%%%%%%%%%%%%%%%%%%%%%%%%%%%%%%%%%%
\section{Preliminaries}
\label{sec:preliminaries}
Let $K$ be arbitrary commutative differential field of characteristic zero with commuting derivations $\p_1, \dots, \p_n$. 
%which we like to denote $\p_{x_1}, \dots, \p_{x_n}$ and by $\p_x$, $\p_y$, $\p_z$ in some examples.
We consider $K$ to be differentially closed. 
We denote by $\mathcal{D}(K)$ the corresponding algebra of differential operators over $K$. 
For any integral vector
$\vecc{v}=(v_1,\dots,v_n) \in \N^n_0$~\footnote{$\N_0$ is the set of natural numbers with zero}
we write $\p^v$ for the differential monomial $\p_{1}^{v_1} \dots \p_{n}^{v_n}$.
Many concepts below are conveniently expressed if we treat the usual multi-indices used to denote derivatives as vectors. In particular we will be using the standard basis vectors $\vecc{e}_i = (0, \dots, 1, \dots 0)$.

%For every vector $\vecc{v}=(v_1,\dots,v_n)$, where $v_i$ are nonnegative integers, %we write 
%operators using multi-index notation as sums of terms of the form $\p^{\vecc{v}}=\p_{x_1}^{v_1} \dots \p_{x_n}^{v_n}$.
%, and so elements of $\mathcal{D}(K)$ are of the form
%$a_\vecc{v} \p^{\vecc{v}}$.
%, where $\vecc{v}$ is an %$n$-long vector of nonnegative integers.

There is a natural partial order on these vectors, where we write $\vecc{u} \preceq \vecc{v}$ iff every entry of $\vecc{u}$ is less than or equal to the corresponding entry of $\vecc{v}$. 
If $\vecc{u} \preceq \vecc{v}$ and $\vecc{u} \neq \vecc{v}$, we say that $\vecc{u}$ is {\em below} $\vecc{v}$, or that $\vecc{v}$ is {\em above} $\vecc{u}$. We extend this terminology to terms of an operator $L \in \mathcal{D}(K)$, saying that $a_{\vecc{u}} \p^{\vecc{u}}$ is below $a_{\vecc{v}} \p^{\vecc{v}}$ iff $\vecc{u}$ is below $\vecc{v}$.
We define the {\em order} of $\vecc{v}$ as the number 
$\sum_{i=1}^n v_i $, and likewise extend this terminology to ``differential monomials''.  
Thus a constant term is of order $0$, a term with a single derivative is of order $1$, and so on. 

The \emph{principal symbol} of an operator is the sum of the highest-order terms. We also define the {\emph{leading part}} of operator $L=\sum a_v \p^v$ as 
\begin{equation}
 Lead_{\preceq} (L) = \sum_{\p^\vecc{v} \in \max_{\preceq}(L)} a_\vecc{v} \p^{\vecc{v}}
  \, ,
\end{equation}
where $\max_{\preceq}(L)$ denotes the set of all
differential monomials $\p^{\vecc{v}}$ of $L$
that are below no other differential monomial in $L$.
Note that the notions of principal symbol and the leading part are not the same. We will call elements of 
$\max_{\preceq}(L)$
%\begin{definition} 
%$\p^{\vecc{v}}$ is 
{\em maximal}.
We call a vector $\vecc{v}$ {\em maximal} iff the differential monomial $\p^{\vecc{v}}$ is maximal.

\begin{example} Let $L = \p_{xx} + a \p_x + b \p_{y} + c$. Then $\p_{xx}$ is its principal symbol, and $\p_{xx} + b \p_y$ is its leading part.
\end{example}

\begin{definition}
The {\em down set} of a set of terms $T$, written $\downarrow T$, is the set consisting of $T$ together with all terms that are below any term in $T$ relative to $\preceq$.  A set of terms is {\em downward closed} iff it is equal to its own down set.  Given terms $s$ and $t$ with $s \prec t$ but where there is no term $t'$ with $s \prec t' \prec t$, we say that $t$ {\em covers} $s$.  A term that is covered by a maximal term, and covered by nothing but maximal terms, will be called {\em submaximal}.
Definitions in the literature of when an element of a partial order is submaximal vary, but the above definition is best for this paper.
\end{definition}

\begin{example}
%Consider terms of operator $L =\p_x \p_y + \p_y +1$. 
%\cat{Let it be $2$ independent variables: $x$ and $y$.} 
The down set of 
$\p_x \p_y$ is $\{\p_x \p_y, \p_x,\p_y,1 \}$, and the set of terms 
$\{\p_x \p_y, \p_y,1\}$ is not downward closed. 
\end{example}
We will be looking at invariants for operators, or more precisely, for invariants of classes of operators with a given set of maximal terms.  Since the gauge transformation of a term with vector $\vecc{v}$ usually contains terms with every vector below $\vecc{v}$, we restrict our investigation to classes of operators with sets of terms that are downward closed.
So given a set of terms $T$, which may possibly have arbitrary coefficients, we let $\mathcal{L}$ be the set of all operators with terms in the downward closure of $T$.  

\begin{definition}
Let a set of terms $T$ be given, where none is below any of the others in $\preceq$. 
%The coefficients of terms in $T$ may be a mix of either given functions or symbols which can stand for any nonzero function in $K$.  
Let $\mathcal{L}(T)$ be the set of all operators $L$ which have $T$ as their sets of maximal terms. 
A set of operators of this form will be said to be {\em generated by its maximal terms}, or {\em maximally generated}.
\end{definition}

 Let $\mathcal{L}$ be a set of operators that is generated by its maximal terms. 
 Observe that the set of terms that may appear in operators in $\mathcal{L}$ is downward closed, and that $\mathcal{L}$ is closed under gauge transformations.  
 Our problem will be to obtain differential invariants for $\mathcal{L}$.

%\begin{definition}\label{invariant definition}
Given a set $\mathcal{L}$ that is closed under gauge transformations, an expression $I$ in terms of the coefficients of terms of members of $\mathcal{L}$ and their derivatives is a (differential) {\em invariant} of $\mathcal{L}$ iff all elements of $\mathcal{L}$ that are related by a gauge transformation have the same value of $I$.
%\end{definition}
%%%%%%%%%%%%%%%%%%%%%%%%%%%%%%%%%%%%%%%%%%%%%%%%%%%%%%%%%%%%%%%
\section{Complete sets of invariants}
\label{sec:complete}
In this section our goal is to give some sufficient condition for a finite set of invariants to be complete. (In the next section we will turn this into a constructive method.)

A {\em complete} set of invariants is such that whenever two operators agree on them, there is a gauge transformation that relates them:
%, i.e. if   
%\begin{definition}
if $\mathcal{L}$ is a set of operators that is closed under gauge transformations, then the set of invariants $\{ I_1, I_2, \dots I_k \}$ is {\em complete} iff every invariant in the set is equal for two operators $L$ and $L'$, there is a function $g \in K$ so that $L' = e^{-g} L e^g$.
%\end{definition}

%\cat{(This is different from the definition you give in the first paragraph of the paper, and it doesn't seem trivial to show that the two definitions are equivalent?).  Seriously, they're still different.  Can I change the other one, then?}
% Yes, I changed it :)) 
We need to distinguish between various kinds of invariants for 
a maximally generated class of operators $\cL$.
\begin{definition}
Let the class of operators $\cL$ be maximally generated, with set of maximal terms $T$.  The coefficients of maximal terms are (trivial) invariants for $\cL$;  we call these invariants {\em maximal}.
\end{definition}

Recall that we are treating multi-indices as vectors.
%It will be convenient to have a function that returns the vector corresponding to a term.
\begin{definition}
For any term $a_\vecc{v} \p^\vecc{v}$, we
call the multi-index $\vecc{v}$, the {\em vector of} $a_\vecc{v} \p^\vecc{v}$, and write
$\vecc{v}=v(a_\vecc{v} \p^\vecc{v})$.
\end{definition}

Let $\cL$ be a maximally generated class of operators, and let $T$ be its set of maximal terms.
Temporarily make the simplifying assumption that every term in $T$ is of the same order
%degree
$k$.
Consider for the moment a particular operator $L$ in $\cL$.  
When $L'$ is obtained from $L$ by a gauge transformation, we have $L' = e^{-g} L e^g$ for some $g \in K$.

Now look at a particular term $a \p_{\vecc{v}}$ in $L$ of degree $k-1$. 
It is covered in the partial order $\preceq$ by some of the maximal terms in $T$ .  
We have that the vectors of maximal terms covering $a \p_{\vecc{v}}$ are of the form $\vecc{v} + \vecc{e}_i$ for $i$ in some subset $S$ of $\{ 1,2, \dots n \}$. 
Now let $a' \p_{\vecc{v}}$ be the term corresponding to $a \p_{\vecc{v}}$ in $L'$.
We have that $a'-a$ is given by
$$a'-a = \sum_{i \in S} (\vecc{v}(i) + 1) b_i g_{x_i}$$
where $b_i$ is the coefficient of the maximal term in $L$ with vector  $\vecc{v} + \vecc{e}_i$.

\begin{example} \label{submaximal_example}
Let $n = 3$, write $x$ for $x_1$, $y$ for $x_2$, and $z$ for $x_3$.  
Let 
\begin{equation*}
    T = \{ p \p_{xxyyz}, q \p_{xyyyz}, \p_{xyzzz} \} 
\end{equation*}
%be a set of three terms of degree $5$, 
where $p, q \in K$.
%and $q$ be symbols for arbitrary functions .
Then the term $a_{121} \p_{xyyz}$ is covered by $p \p_{xxyyz}$ and $q \p_{xyyyz}$ in $\preceq$, but is not below $\p_{xyzzz}$.
In this case, $S = \{ 1,2 \}$, and we have $a_{121}' = a_{121} + 2p g_x + 3q g_y$.
%The term $a \p_{xyyz}$ corresponds to the vector $( 1,2,1 )$.

Other terms that are covered by those in $T$ are those with derivative symbols $\p_{xxyz}$, $\p_{xxyy}$, $\p_{yyyz}$, $\p_{xyyz}$, $\p_{xyyy}$, $\p_{yzzz}$, $\p_{xzzz}$ and $\p_{xyzz}$.
The corresponding terms in the operator $L'$ will have $g_x$ in them when their derivative symbols are $\p_{xyyz}$, $\p_{yyyz}$ and $\p_{yzzz}$.  
Similarly, three terms in $L'$ will have $g_y$ in them and three will have $g_z$ in them.
We can rewrite $a_{121}' = a_{121} + 2p g_x + 3q g_y$ as $2p g_x + 3q g_y = a_{121}' - a_{121}$, and view it as a linear equation in the unknowns $g_x$ and $g_y$.
We have similar equations for each of the other 7 terms that are covered by terms in $T$, giving these 8 equations.
\begin{align}
    p g_z &= a_{220}' - a_{220}
    \label{g equation 1}\\
    2p g_y &= a_{211}' - a_{211}\\
    2p g_x + 3q g_y &= a_{121}' - a_{121}\\
    q g_z &= a_{130}' - a_{130}
    \label{g equation 4}\\
    q g_x &= a_{031}' - a_{031}\\
    3g_z &= a_{112}' - a_{112}
    \label{g equation 6}\\  
    g_y &= a_{103}' - a_{103}
    \label{g equation 7}\\
    g_x &= a_{013}' - a_{013}
    \label{g equation 8}
\end{align}

We only need three equations to solve for $g_x$, $g_y$ and $g_z$, so have five ``extra'' equations.
For example, equations (\ref{g equation 1}), (\ref{g equation 4}) and (\ref{g equation 6}) give us 
$(a_{220}' - a_{220})/p = g_z = (a_{112}' - a_{112})/3$ and
$(a_{130}' - a_{130})/q = g_z = (a_{112}' - a_{112})/3$.
Rearranging these gives
$a_{220}'/p - a_{112}'/3 = a_{220}/p - a_{112}/3$ and
$a_{130}'/q - a_{112}'/3 = a_{130}/q - a_{112}/3$, respectively.
This shows that $a_{220}/p - a_{112}/3$ and $a_{130}/q - a_{112}/3$ are invariants.
Similar calculations with expressions for $g_x$ and $g_y$ would yield three more invariants.
We will call invariants like these {\em extra invariants}.

Next we take three equations where we have solved for $g_x$, $g_y$ and $g_z$.  
In the presence of the five extra invariants, it does not matter what they are;  we will obtain an equivalent set of invariants.
So we will use equations (\ref{g equation 6}), (\ref{g equation 7}) and (\ref{g equation 8}).
Concentrating on (\ref{g equation 7}) and (\ref{g equation 8}), the compatibility condition stating partial derivatives are equal gives us 
$a'_{103x} - a_{103x} = g_{xy} = a'_{013y} - a_{013y}$.
Now we rearrange this, putting primed quantities on one side, and get
$a'_{103x} - a'_{013y} = a_{103x} - a_{013y}$, showing that
$a_{103x} - a_{013y}$ is an invariant.
Similar calculations with the other pairs of equations give us two more invariants.
We call invariants of this kind {\em compatibility invariants}.
(In general we have $n$ variables, and get $n(n-1)/2$ compatibility invariants.)
\end{example}

{While all the maximal terms of $\cL$ were at the same degree 
in the above procedure, it is enough to require the following.}
\begin{definition}\label{approximately flat definition}
Consider the class of operators $\cL$, maximally generated by $T$.  
{Let $M$ be the set of maximal terms of $\cL$, and let $S$ be the set of submaximal terms.
Then $\cL$ is {\em approximately flat} iff there are $n$ distinct elements of $S$, $s_1,s_2,\dots s_n$ so that for every $s_i$ there is a maximal term $m_i \in M$ where the vector of $m_i$ is the sum of $\vecc{e}_i$ and the vector of $s_i$.}
\end{definition}
For example, when $T$ is $\{ \p_{xx}, \p_{y} \}$, then $\cL = \{ \p_{xx} + a_{10} \p_x + \p_y + a_{00} \colon a_{10}, a_{00} \in K \}$ is not approximately flat.  
We have that $S$ is only $\{ a_{10} \p_x  \}$
because the constant term is covered by the nonmaximal term $a_{10} \p_x$.
Then there can not be $n = 2$ distinct elements of S.
%But $a_{10} \p_x$ is only covered by $\p_{xx}$, so while we could take $s_1 = a_{10} \p_x$ and $m_1 = \p_{xx}$, there is no way to get $s_2$ and $m_2$.
  
Whenever all the elements of $T$ have the same degree, $\cL$ is approximately flat.
But when $T = \{ \p_{xxy}, \p_{yy} \}$ for example, $\cL$ is still approximately flat, since we may take 
$s_1 = a_{11} \p_{xy}$,
$m_1 = \p_{xxy}$,
$s_2 = a_{20} \p_{xx}$, and
$m_2 = \p_{xxy}$.
See Figure~\ref{approximately flat figure}.

Another example of a class that is not approximately flat is obtained by 
taking $T$ to be $\{ \p_{xz}, \p_{yz} \}$, so $\cL = \{ \p_{xz} + \p_{yz} + a_{100} \p_x + a_{010}\p_y + 
a_{001} \p_z + a_{000} \colon a_{100}, a_{010},a_{001},a_{000} \in K \}$.
Then $M = \{ \p_{xz}, \p_{yz} \}$, while
$S = \{ a_{100}\p_{x}, a_{010}\p_{y}, a_{001} \p_z \}$ is a set with $n = 3$ elements.
But $a_{001} \p_z$ is our only possible choice for both $s_1$ and $s_2$, and we fail to have distinct elements of $S$ for each $i$.

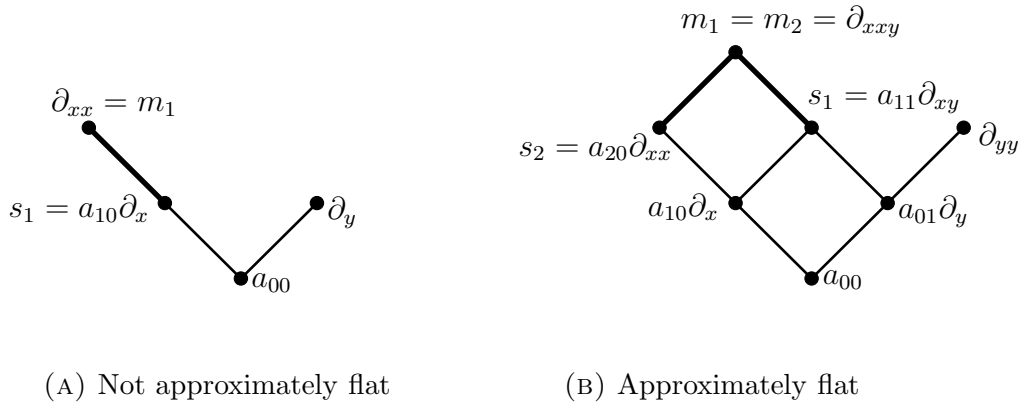
\begin{figure}
\advance\leftskip-1cm
\begin{subfigure}[]{0.4\linewidth}
\begin{tikzpicture}[line cap=round,line join=round,>=triangle 45,x=1cm,y=1cm]
\clip(-3.5,-1) rectangle (3,4);
\draw [line width=2pt] (-2,2)-- (-1,1);
\draw [line width=1pt] (-1,1)-- (0,0);
\draw [line width=1pt] (1,1)-- (0,0);
\draw (-2.6403992896660045,2.6586238644901297) node[anchor=north west] {$\partial_{xx} = m_1$};
\draw (-3.2027447578717287,1.2516161464380866) node[anchor=north west] {$s_1 = a_{10} \partial_x$};
\draw (0.9932322928761687,1.224295608223484) node[anchor=north west] {$\partial_y$};
\draw (-0.003967351956833749,0.21343569428318127) node[anchor=north west] {$a_{00}$};
\begin{scriptsize}
\draw [fill=black] (0,0) circle (2.5pt);
\draw [fill=black] (-1,1) circle (2.5pt);
\draw [fill=black] (-2,2) circle (2.5pt);
\draw [fill=black] (1,1) circle (2.5pt);
\end{scriptsize}
\end{tikzpicture}
\caption{Not approximately flat}
\end{subfigure}
%\hfill
\begin{subfigure}[]{0.4\linewidth}
\begin{tikzpicture}[line cap=round,line join=round,>=triangle 45,x=1cm,y=1cm]
\clip(-4.5,-1) rectangle (3,4);
\draw [line width=1pt] (-2,2)-- (-1,1);
\draw [line width=1pt] (-1,1)-- (0,0);
\draw [line width=1pt] (1,1)-- (0,0);
\draw (-4,2.066143485938711) node[anchor=north west] {$s_2 = a_{20}\partial_{xx}$};
\draw (-2.3,1.2715864461404187) node[anchor=north west] {$a_{10} \partial_x$};
\draw (0.9945080828606878,1.2410265599943306) node[anchor=north west] {$a_{01}\partial_y$};
\draw (0.001311783112822473,0.2478302602464655) node[anchor=north west] {$a_{00}$};
\draw [line width=2pt] (-1,3)-- (-2,2);
\draw [line width=2pt] (-1,3)-- (0,2);
\draw [line width=1pt] (0,2)-- (-1,1);
\draw [line width=1pt] (0,2)-- (1,1);
\draw [line width=1pt] (2,2)-- (1,1);
\draw (2.0488241549007293,2.2036629735961077) node[anchor=north west] {$\partial_{yy}$};
\draw (-1.8628412717985554,3.7622171670466034) node[anchor=north west] {$m_1 = m_2 = \partial_{xxy}$};
\draw (-0.2,2.7537409242256943) node[anchor=north west] {$s_1 = a_{11} \partial_{xy}$};
\begin{scriptsize}
\draw [fill=black] (0,0) circle (2.5pt);
\draw [fill=black] (-1,1) circle (2.5pt);
\draw [fill=black] (-2,2) circle (2.5pt);
\draw [fill=black] (1,1) circle (2.5pt);
\draw [fill=black] (0,2) circle (2.5pt);
\draw [fill=black] (-1,3) circle (2.5pt);
\draw [fill=black] (2,2) circle (2.5pt);
\end{scriptsize}
\end{tikzpicture}
\caption{Approximately flat}
\end{subfigure}
\caption{For Definition \ref{approximately flat definition}.}
\label{approximately flat figure}
\end{figure}
%In general, a maximally generated class $\cL$ fails to be approximately flat iff there is some basis vector $\vecc{e}_i$ so that for every maximal term $m$ where the $i$-th component of  $v(m)$ is at least $1$, the term corresponding to $v(m) - \vecc{e}_i$  is covered by a non-maximal term.

\begin{definition}\label{kinds of invariants}
Let $\cL$ be maximally generated, with set of maximal terms $T$.
Assume that $\cL$ is approximately flat.
The coefficients of terms in $T$ are {\em maximal} invariants of $\cL$, and we denote them by $I_m$.
%(When these coefficients are $1$, we tend not to list them as maximal invariants.)
Invariants obtained by equating expressions for some $g_{x_i}$ are {\em extra} invariants, which we denote by $I_e$.
(Letting $s$ be the number of submaximal terms, there are $s$ different equations where linear combinations of the $g_{x_i}$ are equal to some difference of the form $a_{\vecc{v}}' - a_{\vecc{v}}$, so there will be $s-n$ many extra invariants.)
If we have expressions $E_i$ and $E_j$ with $g_{x_i} = E_i$ and $g_{x_j} = E_j$ and obtain an invariant by rearranging the equation 
$E_{ix_j} = E_{jx_i}$, this invariant is a {\em compatibility} invariant, which we will denote $I_c$.

Finally, suppose that $a_{\vecc{v}} \p_\vecc{v}$ is a term that is neither maximal or submaximal, and that $E$ is an expression only involving the coefficients of terms that are above $a_{\vecc{v}} \p_\vecc{v}$, or that are maximal or submaximal (and their derivatives).  
Then an invariant of the form $a_{\vecc{v}} - E$ is an {\em upward} invariant, we call it an upward invariant for $a_{\vecc{v}} \p_\vecc{v}$, and denote it $I_\vecc{v}$.
\end{definition}

Examples of all of these types of invariants appear in the literature.  
%Note that extra invariants are independent of each other by construction.
%and compatibility
%One would take the right number of extra invariants, and a compatability invariant for every pair of variables.
%Upward invariants for different terms would have to be independent.
%For simplicity, assume we have two upward invariants, $I_u$ for $a_u \p^u$ and $I_v$ for $a_v \p^v$, where the vector $u$ is not above the vector $v$ in the partial order $\prec$.
%(So $u$ is lower than, or anyway, not higher than $v$.)
%Then $I_v$ does not contain $a_u$, while $I_u$ does.
%This implies independence.
%We just need names for them.

\begin{theorem}\label{sufficient_for_a_complete_set_theorem}
%Assume that our differential field $K$ is closed under.
Let $\cL$ be maximally generated, and let $T$ be its set of maximal terms.
Assume that $\cL$ is approximately flat.
Suppose that $I$ is a set of invariants for $\cL$ so that the following hold:
\begin{enumerate}
    \item $I$ contains all the maximal invariants of $\cL$.
    \item $I$ contains $s-n$ extra invariants, where $s$ is the number of submaximal terms.
    \item $I$ contains $n(n-1)/2$ compatibility invariants, one for each
    %~\footnote{This will make them independent of each other.}  
    possible second-order compatibility partial of $g$.
    \item $I$ contains an upward invariant for every term that is not maximal or submaximal.
\end{enumerate}
Then the set of invariants $I$ is complete.
\end{theorem}
\begin{proof}
Let $I$ be a set of invariants as above.
Let $L \in \cL$ be given, and let $L'$ be an element of $\cL$ where $L$ and $L'$ have the same values for all invariants in $I$.
We must show there exists a function $g \in K$ so that $L' = e^{-g} L e^g$.

Consider solving for the derivatives of $g$ in $L'$.  
Since $\cL$ is approximately flat, we only obtain linear equations in the first partial derivatives of $g$.
The values of these $g_{x_i}$ are the same no matter which equations we use, because $I$ contains enough extra invariants.
Since $I$ contains enough compatibility invariants, we have a (compatible) system of first order partial differential equations for $g$. %Since $K$ is closed
%\cat{(under what?  If this issue is traditionally ignored, then this part has to go too.)}
This gives a value for $g$ in $K$ that is unique up to an additive constant, and the additive constant does not change what our candidate gauge transformation is.

So we have $g \in K$ where $L'' = e^{-g} L e^g$ agrees with $L'$ on the coefficients of all maximal and submaximal terms.
It remains to show that $L''$ and $L'$ agree on their remaining terms.
We do this by downward induction on the degree of terms.
For our basis, let $m$ be the highest degree of a term in $\cL$ that is not maximal or submaximal, and consider any term $a_{\vecc{v}} \p^{\vecc{v}}$ of degree $m$.
If $a_{\vecc{v}} \p^{\vecc{v}}$ is maximal or submaximal, we already have that it has the same value in $L'$ and $L''$.
So suppose $a_{\vecc{v}} \p^{\vecc{v}}$ is not maximal or submaximal.
By our choice of $m$, every term of $\cL$ that is above $a_{\vecc{v}} \p^{\vecc{v}}$ is submaximal or maximal, and all of these terms have the same value in $L'$ and $L''$. 
Now $a_{\vecc{v}} \p^{\vecc{v}}$ has an upward invariant of the form $a - E$, where $E$ is an expression only involving terms above $a_{\vecc{v}} \p^{\vecc{v}}$.
Since $L'$ and $L$ have the same values for all invariants in $I$, this upward invariant is equal in $L'$ and $L$.
And since $L''$ is obtained from $L$ by a gauge transformation, the invariant has the same value in $L$ and in $L''$, implying that it has the same value in $L'$ and $L''$.
Since $E$ and the invariant both have the same value in $L'$ and $L''$, we have that $a$ has the same value in $L'$ and $L''$.

The inductive argument now continues.
In the next stage, terms of degree $m-1$ are only below terms that have the same value in $L'$ and $L''$, and are either maximal, submaximal, or have upward invariants.  
In any event, all terms of degree $m-1$ have the same value in $L'$ and $L''$.
The process continues, eventually showing that the term of degree $0$ has the same value in $L'$ and $L''$.
\end{proof}

%%%%%%%%%%%%%%%%%%%%%%%%%%%%%%%%%%%%%%%%%%%%%%%%%%%%%%%%%%%%%%
\section{Constructing complete sets of invariants}
\label{sec:method}
%\cat{Roughly our goal is to construct a complete set of gauge invariants of special type: they are either have a short formula or are connected with some generalized version of incomplete factorization of the operator.}
%We work toward a constructive proof of the existence of a complete set of invariants under gauge 
%transformations.  
%Thur method is an extension of one used in~\cite{obstacle2}, where
%invariants were produced from partial factorizations.  
%Our key tool is the following.

\begin{theorem}\label{unique_C_theorem}  
Let $\mathcal{L}$, $\mathcal{C}$ and $\mathcal{N}$ be classes 
of operators that are closed under gauge transformations. Assume that 
for every $L \in \mathcal{L}$, there is a unique $C \in \mathcal{C}$ so 
that $N = L-C$ is in $\mathcal{N}$.  Then all of the 
%coefficients of the  principal symbol
invariants of $\cN$ are invariants for $\mathcal{L}$.
\end{theorem}
\begin{proof}
Let $\cL$, $\mathcal{C}$ and $\mathcal{N}$ be closed under gauge transformations, and assume that for each $L \in \cL$ there is a unique $C \in \mathcal{C}$ with $L-C \in \mathcal{N}$.  
That is, every $L \in \cL$ determines a unique $N \in \mathcal{N}$.
Now let a particular $L \in \cL$ be given, where $N = L-C$ is in $\mathcal{N}$.
Gauging by any nonzero $g \in K$, we have $N' = g^{-1} N g \in \mathcal{N}$, $L' = g^{-1} L g \in \cL$ and $C' = g^{-1} C g \in \mathcal{C}$ with $N' = L'-C'$.  
By uniqueness, $N'$ must be the element of $\mathcal{N}$ determined by $L'$.
The 
%coefficients of the principal symbol 
invariants of $N' = g^{-1} N g$ are the same as the corresponding 
%coefficients 
invariants in $N$, making them invariants of $\cL$.
\end{proof}

In an application,
$\mathcal{L}$ would be the class of all operators of 
a particular form, and the classes $\mathcal{C}$ and $\mathcal{N}$ would 
be tailored to produce a family of invariants for $\mathcal{L}$.  
The invariants of $\cN$ that we will use are usually the coefficients of its maximal terms.

%Note that it is critical that each $L \in \cL$ determine a unique ``remainder'' $N \in \mathcal{N}$.  
%Even when it is plausible that there is such a unique $N$, we must be able to demonstrate uniqueness.
%The usual reason we can do this is when all the equations for unknown coefficients in $C$ are linear.

\begin{example} \label{simple example}
Suppose $n = 2$ and let $\mathcal{L}$ 
be the class of operators maximally generated by $T$, where
\begin{equation*}
    T = \{ \p_{xxy} \}. 
\end{equation*}
%is the set of operators with principal 
%symbol $\p_{xxy}$, where the coefficients of $\p_x^m \p_y^n$ are zero 
%unless $m \leq 2$ and $n \leq 1$.
Consider the operator 
$$L = \p_{xxy} + a_{20}\p_{xx} + a_{11}\p_{xy} + a_{10}\p_{x}
 + a_{01}\p_{y}  + a_{00}  \in \mathcal{L} \ .
 $$
 
To obtain a complete set of invariants, we first include all the needed maximal, extra, and compatibility invariants.
The coefficient 1 of $\p_{xxy}$ is a maximal invariant.
%, since it is the coefficient of a maximal term.  
There are two submaximal terms, $a_{20}\p_{xx}$ and $a_{11}\p_{xy}$.  
Since the dimension is $n = 2$, there are no extra invariants.
There is one compatibility invariant, which we get by observing that $e^{-g} L e^g = L' = \p_{xxy} + a'_{20}\p_{xx} + a'_{11}\p_{xy} + \dots$
has $a'_{20} = a_{20} + g_y$ and $a'_{11} = a_{11} + 2g_x$.
Then solving and differentiating, $(2a'_{20} - 2a_{20})_x = 2g_{xy} = (a'_{11} - a_{11})_y$.
%Rearranging, we get $2a'_{20}_y - a'_{11}_x = 2a_{20}_y - a_{11}_x$.
This shows that
$$I_c = 2a_{20y} - a_{11x}$$ is the desired compatibility invariant.

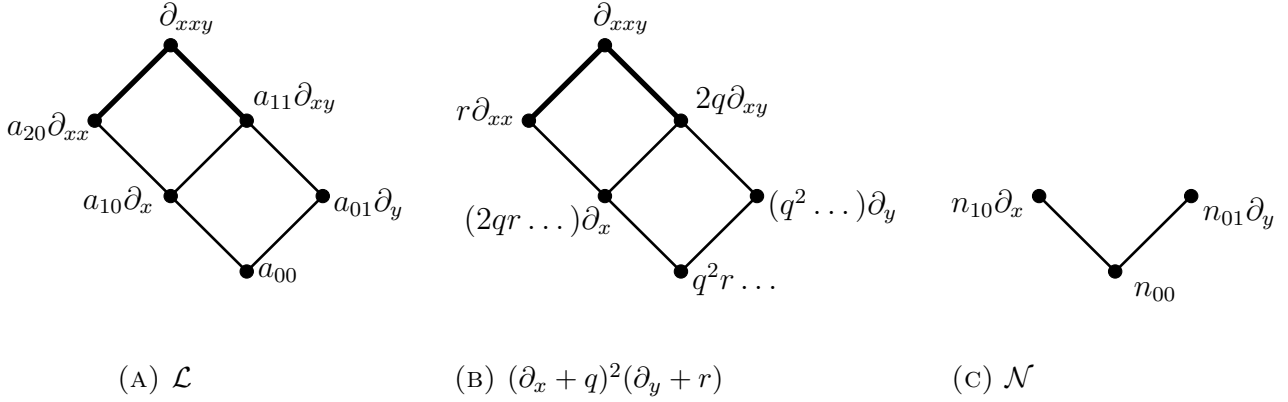
\begin{figure}
\advance\leftskip-2.5cm
\begin{subfigure}{0.35\linewidth}
\begin{tikzpicture}[line cap=round,line join=round,>=triangle 45,x=1cm,y=1cm]
\clip(-4,-1) rectangle (4,4);
\draw [line width=1pt] (-2,2)-- (-1,1);
\draw [line width=1pt] (-1,1)-- (0,0);
\draw [line width=1pt] (1,1)-- (0,0);
\draw (-3.3,2.2495028028152397) node[anchor=north west] {$a_{20}\partial_{xx}$};
\draw (-2.3,1.2715864461404187) node[anchor=north west] {$a_{10} \partial_x$};
\draw (0.9945080828606878,1.2410265599943306) node[anchor=north west] {$a_{01}\partial_y$};
\draw (0.001311783112822473,0.2478302602464655) node[anchor=north west] {$a_{00}$};
\draw [line width=2pt] (-1,3)-- (-2,2);
\draw [line width=2pt] (-1,3)-- (0,2);
\draw [line width=1pt] (0,2)-- (-1,1);
\draw [line width=1pt] (0,2)-- (1,1);
\draw (-1.2974833780959245,3.7) node[anchor=north west] {$\partial_{xxy}$};
\draw (-0.05980798917935386,2.65) node[anchor=north west] {$a_{11} \partial_{xy}$};
\begin{scriptsize}
\draw [fill=black] (0,0) circle (2.5pt);
\draw [fill=black] (-1,1) circle (2.5pt);
\draw [fill=black] (-2,2) circle (2.5pt);
\draw [fill=black] (1,1) circle (2.5pt);
\draw [fill=black] (0,2) circle (2.5pt);
\draw [fill=black] (-1,3) circle (2.5pt);
\end{scriptsize}
\end{tikzpicture}
\caption{$\cL$}
\end{subfigure}
%\hfill
\begin{subfigure}[]{0.35\linewidth}
\begin{tikzpicture}[line cap=round,line join=round,>=triangle 45,x=1cm,y=1cm]
\clip(-4,-1) rectangle (4,4);
\draw [line width=1pt] (-2,2)-- (-1,1);
\draw [line width=1pt] (-1,1)-- (0,0);
\draw [line width=1pt] (1,1)-- (0,0);
\draw (-3.1,2.4481420627648127) node[anchor=north west] {$r\partial_{xx}$};
\draw (-3,1.0271073569717135) node[anchor=north west] {$(2qr\dots) \partial_x$};
\draw (0.9945080828606878,1.2410265599943306) node[anchor=north west] {$(q^2 \dots)\partial_y$};
\draw (0.001311783112822473,0.2478302602464655) node[anchor=north west] {$q^2 r \dots$};
\draw [line width=2pt] (-1,3)-- (-2,2);
\draw [line width=2pt] (-1,3)-- (0,2);
\draw [line width=1pt] (0,2)-- (-1,1);
\draw [line width=1pt] (0,2)-- (1,1);
\draw (-1.2974833780959245,3.7) node[anchor=north west] {$\partial_{xxy}$};
\draw (0.047151612331954716,2.6) node[anchor=north west] {$2q \partial_{xy}$};
\begin{scriptsize}
\draw [fill=black] (0,0) circle (2.5pt);
\draw [fill=black] (-1,1) circle (2.5pt);
\draw [fill=black] (-2,2) circle (2.5pt);
\draw [fill=black] (1,1) circle (2.5pt);
\draw [fill=black] (0,2) circle (2.5pt);
\draw [fill=black] (-1,3) circle (2.5pt);
\end{scriptsize}
\end{tikzpicture}
\caption{$(\p_x + q)^2 (\p_y + r)$}
\end{subfigure}
%\hfill
\begin{subfigure}[]{0.3\linewidth}
\begin{tikzpicture}[line cap=round,line join=round,>=triangle 45,x=1cm,y=1cm]
\clip(-4,-1) rectangle (4,4);
\draw [line width=1pt] (-1,1)-- (0,0);
\draw [line width=1pt] (1,1)-- (0,0);
\draw (-2.3,1.2257466169212865) node[anchor=north west] {$n_{10} \partial_x$};
\draw (1.009788025933732,1.118787015409978) node[anchor=north west] {$n_{01} \partial_y$};
\draw (0.1,0) node[anchor=north west] {$n_{00}$};
\begin{scriptsize}
\draw [fill=black] (0,0) circle (2.5pt);
\draw [fill=black] (-1,1) circle (2.5pt);
\draw [fill=black] (1,1) circle (2.5pt);
\end{scriptsize}
\end{tikzpicture}
\caption{$\cN$}
\end{subfigure}
\caption{For Example \ref{simple example}}
\label{simple example figure}
\end{figure}
We now let
\begin{equation*}
    \mathcal{C} = \{ (\p_x 
+ q)^2 (\p_y + r) \colon q,r \in K \} \ ,
\end{equation*}
%where $K$ is our differential 
%field, 
and take $\mathcal{N}$ to be the set of elements of $\mathcal{L}$ 
with coefficients of $\p_{xxy}$, $\p_{xy}$ and $\p_{xx}$ all 
zero.  This will give the first batch of upward invariants.  
(See Figure \ref{simple example figure}.)

%For 
%$L = \p_{xxy} + a_{20} \p_{xx} + a_{11} \p_{xy} + 
%a_{10} \p_x + a_{01} \p_y + a_{00} \in \cL$ and 
%\begin{equation*}
%    C = (\p_x + q)^2 (\p_y + r) \in \cC %\ ,
%\end{equation*}
We have that
$L-C$ is
%$(\p_{xxy} + a_{20} \p_{xx} + a_{11} \p_{xy} +
%a_{10} \p_x + a_{01} \p_y + a_{00}) - 
%(\p_{xxy} + r \p_{xx} + 2q \p_{xy} + 
%(2qr + 2r_x) \p_x + (q^2 + q_x) \p_y + 
%((q^2 + q)r + 2qr_x + r_{xx})) =
$(a_{20} - r) \p_{xx} + (a_{11} - 2q) \p_{xy} + 
(a_{10} - (2qr + 2r_x)) \p_x + (a_{01} - (q^2 + q_x)) \p_y 
+ (a_{00} - ((q^2 + q)r + 2qr_x + r_{xx}))$.

For this to be in $\cN$, we must have 
$a_{20} - r = a_{11} - 2q = 0$, giving 
$r = a_{20}$ and $q = a_{11} / 2$, which uniquely determines
$C$.  Substituting these in, we get that 
$N = L - C =
%(a_{10} - (2qr + 2r_x)) \p_x + 
%(a_{01} - (q^2 + q_x)) \p_y 
%+ (a_{00} - ((q^2 + q)r + 2qr_x + %r_{xx}))
%=
(a_{10} - (a_{11} a_{20} + 2 a_{20x})) \p_x + 
(a_{01} - (a_{11}^2/4 + a_{11x}/2)) \p_y 
+ (a_{00} - ((a_{11}^2/4 + a_{11}/2) a_{20} + a_{11} a_{20x} + a_{20xx}))$.
The coefficients of $\p_x$ and $\p_x$ are two invariants,
\begin{align*}
I_{10} =& a_{10} - (a_{11} a_{20} + 2 a_{20x}) \mbox{\quad and}\\
I_{01} =& a_{01} - (a_{11}^2/4 + a_{11x}/2).
\end{align*}

In particular, we get the associated %incomplete factorization
representation
\begin{equation}
\label{ex4:incomplete:1}
L = (\p_x 
+ q)^2 (\p_y + r) + I_{10} \p_x + I_{01} \p_y + s 
\end{equation}
for some $q,r,s \in K$.
This gives meaning to the invariants.  We have that $I_{10}$ is zero iff there is a representation of $L$ without a $\p_x$ term, and so on.

Next we add more terms to the form of $\mathcal{C}$, 
%for instance 
setting 
\begin{equation*}
\mathcal{C}' = \{ (\p_x + q)^2 (\p_y + r) + 
(\p_x + s)(\p_y + t) \colon q,r,s,t \in K \}
\end{equation*}
and letting 
$\mathcal{N}'$ be the set of elements of $\mathcal{N}$ with coefficients 
of $\p_x$, and $\p_y$ both zero.  
Letting $C' \in \cC'$, we have that
$L-C'$ is 
%$(\p_{xxy} + a_{20} \p_{xx} + a_{11} \p_{xy} +
%a_{10} \p_x + a_{01} \p_y + a_{00}) - 
%(\p_{xxy} + r \p_{xx} + 2q \p_{xy} + 
%(2qr + 2r_x) \p_x + (q^2 + q_x) \p_y + 
%((q^2 + q)r + 2qr_x + r_{xx})
%+ \p_{xy} + t \p_x + s \p_y + (st + t_x))
%=
$
(a_{20} - r) \p_{xx} + (a_{11} - 1 - 2q) \p_{xy} + 
(a_{10} - (2qr + 2r_x) - t) \p_x + 
(a_{01} - (q^2 + q_x) - s) \p_y 
+ (a_{00} - ((q^2 + q)r + 2qr_x + r_{xx} + (st + t_x)))$.

There is one way to make this be an element of $\cN'$.
As before, we let $r = a_{20}$.  
With a slight change, we let $q = (a_{11}-1)/2$.
Now that $q$ and $r$ are determined, we let
$s = a_{01} - (q^2 + q_x)$ and
$t = a_{10} - (2qr + 2r_x)$.

Then the constant term of $N' = L - C'$ is 
$a_{00} - ((q^2 + q)r + 2qr_x + r_{xx} + (st + t_x))$.
It is an invariant of $\cL$.  
Expanding, we get 

\begin{align*}
\begin{split}
I_{00} =& a_{00} - ((a_{11}-1)^2/4 + (a_{11}-1)/2)a_{20}
- (a_{11}-1)a_{20x} - 
a_{20xx} -\\ 
&(a_{01} - ((a_{11}-1)/2)^2 + (a_{11}-1)/2)_x)(a_{10} - ((a_{11}-1))a_{20} + 2a_{20x}) \\
&- (a_{10} - (a_{11}-1)a_{20} - 2a_{20x})_x 
\end{split}
\end{align*}

In particular, we get the associated representation,
\begin{equation*}
 L =  (\p_x + q)^2 (\p_y + r) + 
(\p_x + s)(\p_y + t) + I_{00}
\end{equation*}
for some $q,r,s,t \in K$, where 
$q$ and $r$ are possibly different from $q$ and $r$ in incomplete factorization~\eqref{ex4:incomplete:1}.

%There is no great gain to writing an invariant like this in expanded form.  In the future, we will write more complicated invariants by giving a series of simpler equations, leaving the sequential substitutions for the reader.
\end{example}

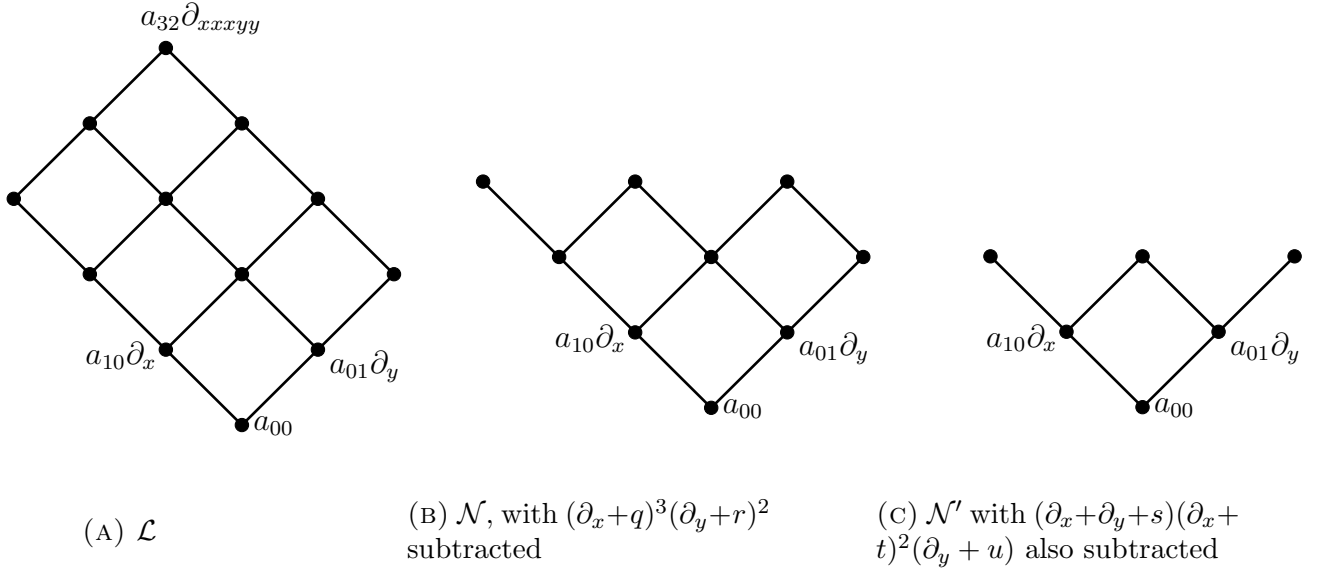
\begin{figure}
\advance\leftskip-1.2cm
\begin{subfigure}[]{0.3\linewidth}
\begin{tikzpicture}[line cap=round,line join=round,>=triangle 45,x=1cm,y=1cm]
\clip(-4,-1) rectangle (3,6);
\draw [line width=1pt] (-1,1)-- (0,0);
\draw [line width=1pt] (1,1)-- (0,0);
\draw (-2.2,1.2406712124810024) node[anchor=north west] {$a_{10} \partial_x$};
\draw (1.0083666358804297,1.1312241783764077) node[anchor=north west] {$a_{01} \partial_y$};
\draw (0.007708038352705493,0.25564790553965033) node[anchor=north west] {$a_{00}$};
\draw [line width=1pt] (-2,2)-- (-1,1);
\draw [line width=1pt] (0,2)-- (1,1);
\draw [line width=1pt] (2,2)-- (1,1);
\draw [line width=1pt] (0,2)-- (-1,1);
\draw [line width=1pt] (-1,3)-- (-2,2);
\draw [line width=1pt] (-1,3)-- (0,2);
\draw [line width=1pt] (1,3)-- (2,2);
\draw [line width=1pt] (1,3)-- (0,2);
\draw [line width=1pt] (0,4)-- (-1,3);
\draw [line width=1pt] (0,4)-- (1,3);
\draw [line width=1pt] (-3,3)-- (-2,2);
\draw [line width=1pt] (-2,4)-- (-1,3);
\draw [line width=1pt] (-2,4)-- (-3,3);
\draw [line width=1pt] (-1,5)-- (-2,4);
\draw [line width=1pt] (-1,5)-- (0,4);
\draw (-1.4776445673525103,5.743634901355755) node[anchor=north west] {$a_{32} \partial_{xxxyy}$};
\begin{scriptsize}
\draw [fill=black] (0,0) circle (2.5pt);
\draw [fill=black] (-1,1) circle (2.5pt);
\draw [fill=black] (1,1) circle (2.5pt);
\draw [fill=black] (-2,2) circle (2.5pt);
\draw [fill=black] (0,2) circle (2.5pt);
\draw [fill=black] (2,2) circle (2.5pt);
\draw [fill=black] (-3,3) circle (2.5pt);
\draw [fill=black] (-1,3) circle (2.5pt);
\draw [fill=black] (1,3) circle (2.5pt);
\draw [fill=black] (0,4) circle (2.5pt);
\draw [fill=black] (-2,4) circle (2.5pt);
\draw [fill=black] (-1,5) circle (2.5pt);
\end{scriptsize}
\end{tikzpicture}
\caption{$\cL$}
\end{subfigure}
\hfill
\begin{subfigure}[]{0.3\linewidth}
\begin{tikzpicture}[line cap=round,line join=round,>=triangle 45,x=1cm,y=1cm]
\clip(-4,-1) rectangle (3,6);
\draw [line width=1pt] (-1,1)-- (0,0);
\draw [line width=1pt] (1,1)-- (0,0);
\draw (-2.2,1.2406712124810024) node[anchor=north west] {$a_{10} \partial_x$};
\draw (1.0083666358804297,1.1312241783764077) node[anchor=north west] {$a_{01} \partial_y$};
\draw (0.007708038352705493,0.25564790553965033) node[anchor=north west] {$a_{00}$};
\draw [line width=1pt] (-2,2)-- (-1,1);
\draw [line width=1pt] (0,2)-- (1,1);
\draw [line width=1pt] (2,2)-- (1,1);
\draw [line width=1pt] (0,2)-- (-1,1);
\draw [line width=1pt] (-1,3)-- (-2,2);
\draw [line width=1pt] (-1,3)-- (0,2);
\draw [line width=1pt] (1,3)-- (2,2);
\draw [line width=1pt] (1,3)-- (0,2);
\draw [line width=1pt] (-3,3)-- (-2,2);
\begin{scriptsize}
\draw [fill=black] (0,0) circle (2.5pt);
\draw [fill=black] (-1,1) circle (2.5pt);
\draw [fill=black] (1,1) circle (2.5pt);
\draw [fill=black] (-2,2) circle (2.5pt);
\draw [fill=black] (0,2) circle (2.5pt);
\draw [fill=black] (2,2) circle (2.5pt);
\draw [fill=black] (-3,3) circle (2.5pt);
\draw [fill=black] (-1,3) circle (2.5pt);
\draw [fill=black] (1,3) circle (2.5pt);
\end{scriptsize}
\end{tikzpicture}
\caption{$\cN$, with $(\p_x +q)^3 (\p_y + r)^2$ subtracted}
\end{subfigure}
\hfill
\begin{subfigure}[]{0.3\linewidth}
\begin{tikzpicture}[line cap=round,line join=round,>=triangle 45,x=1cm,y=1cm]
\clip(-3.5,-1) rectangle (2.5,6);
\draw [line width=1pt] (-1,1)-- (0,0);
\draw [line width=1pt] (1,1)-- (0,0);
\draw (-2.2,1.2406712124810024) node[anchor=north west] {$a_{10} \partial_x$};
\draw (1.0083666358804297,1.1312241783764077) node[anchor=north west] {$a_{01} \partial_y$};
\draw (0.007708038352705493,0.25564790553965033) node[anchor=north west] {$a_{00}$};
\draw [line width=1pt] (-2,2)-- (-1,1);
\draw [line width=1pt] (0,2)-- (1,1);
\draw [line width=1pt] (2,2)-- (1,1);
\draw [line width=1pt] (0,2)-- (-1,1);
\begin{scriptsize}
\draw [fill=black] (0,0) circle (2.5pt);
\draw [fill=black] (-1,1) circle (2.5pt);
\draw [fill=black] (1,1) circle (2.5pt);
\draw [fill=black] (-2,2) circle (2.5pt);
\draw [fill=black] (0,2) circle (2.5pt);
\draw [fill=black] (2,2) circle (2.5pt);
\end{scriptsize}
\end{tikzpicture}
\caption{$\cN'$ with $(\p_x + \p_y + s)(\p_x +t)^2 (\p_y +u)$ also subtracted}
\end{subfigure}
\caption{For Example \ref{complicated example}}
\label{complicated example figure}
\end{figure}

We will now work through a more complicated example in less detail, commenting on the process as we go. 
\begin{example} \label{complicated example}
Suppose $n = 2$ and 
let $\mathcal{L}$ 
be the class of operators maximally generated by $T$, where
\begin{equation*}
    T = \{ \p_{xxxyy} \}. 
\end{equation*}
%$\mathcal{L}$ is the set of operators with principal 
%symbol $\p_{xxxyy}$, where the coefficients of $\p_x^m \p_y^n$ are zero 
%unless $m \leq 3$ and $n \leq 2$.  
So elements of $\cL$ have the form
%\begin{equation*}
%\cL =    
$\p_{xxxyy} + a_{31} \p_{xxxy} + a_{22} \p_{xxyy} +
a_{30} \p_{xxx} + a_{21} \p_{xxy} + a_{12} \p_{xyy} + a_{20} \p_{xx} + a_{11} \p_{xy} + a_{02} \p_{yy} + a_{10} \p_x + a_{01} \p_y + a_{00}$.
%\end{equation*}
%Similarly to the previous example, we get that
The coefficient of $\p_{xxxyy}$ is the maximal invariant and there are no extra invariants.
Letting primes denote the gauge
action on the coefficients,
%transformation as before
we have $a'_{22} = a_{22} + 3g_x$ and $a'_{31} = a_{31} + 2g_y$.
Then $2a'_{22} - 2a_{22} = 6g_x$ and $3a'_{31} - 3a_{31} = 6g_y$.
Thus $2a'_{22y} - 2a_{22y} = 6g_{xy}$ and $3a'_{31x} - 3a_{31x} = 6g_{xy}$.
So $2a'_{22y} - 2a_{22y} = 3a'_{31x} - 3a_{31x}$, and
$$I_c=2a_{22y} - 3a_{31x}$$ is the desired compatibility invariant.

We now let
\begin{equation*}
\mathcal{C} = \{ (\p_x 
+ q)^3 (\p_y + r)^2 \colon q,r \in K \} \ ,    
\end{equation*}
and take $\mathcal{N}$ to be the set of elements of $\mathcal{L}$ 
with coefficients of $\p_{xxxyy}$, $\p_{xxxy}$ and $\p_{xxyy}$ all 
zero.  
We see that $(\p_x + q)^3 (\p_y + r)^2$
is $\p_{xxxyy} + 2r \p_{xxxy} + 3q \p_{xxyy} + (r^2 +r_y)\p_{xxx} + 6(qr + r_x)\p_{xxy} + 3(q^2+q_x)\p_{xyy} + \dots $, so for 
$L = \p_{xxxyy} + a_{31} \p_{xxxy} + a_{22} \p_{xxyy} + \dots $ there is a unique choice of $q$ and $r$ to get $L - (\p_x + q)^3 (\p_y + r)^2 \in \cN$, 
we let $r = a_{31}/2$ and $q = a_{22}/3$.
This gives the first batch of three upward invariants, for $a_{30}$, $a_{21}$ and $a_{12}$,
$$I_{30} = a_{30} - (r^2 + r_y) \ .$$

%So far, this is just 
%the method of partial factorizations.

Now we add more terms to the form of $\mathcal{C}$, 
%for instance 
setting
\begin{equation*}
\mathcal{C}' = \{ (\p_x + q)^3 (\p_y + r)^2 + (\p_x + \p_y + 
s)(\p_x + t)^2 (\p_y + u) \colon q,r,s,t,u \in K \}    
\end{equation*}
and letting 
$\mathcal{N}'$ be the set of elements of $\mathcal{N}$ with coefficients 
of $\p_{xxx}$, $\p_{xxy}$ and $\p_{xyy}$ all zero.  

One has to be careful in the choice of $\mathcal{C}'$.
%in order 
%to make the method always work.  
Intuitively, our choice of the additional term 
$(\p_x + \p_y + s)(\p_x + t)^2 (\p_y + u)$ was good because it had three free parameters, precisely the number of terms we were trying to ``zero out'' when going to $\mathcal{N'}$.
While it is fine to use a factor such as $(\p_x + \p_y + s)$, this is not necessary.

Note that the added term of 
$(\p_x + \p_y + s)(\p_x + t)^2 (\p_y + u)$
has principal symbol $\p_{xxxy} + \p_{xxyy}$.
If we were to use the same values of $q$ and $r$ as in the previous step, any element of 
$\cL - \cC'$ would have coefficients of 
$\p_{xxxy}$ and $\p_{xxyy}$ that were $-1$.
So we modify $q$ and $r$ to make these coefficients zero.  We already have that for any $L \in \cL$, $q$ and $r$ can be chosen to make the coefficients of $\p_{xxxy}$ and $\p_{xxyy}$ zero in $L - \cC$, so we merely
use $L - (\p_{xxxy} + \p_{xxyy}) \in \cL$
to determine our new $q$ and $r$.

Since $(\p_x + \p_y + s)(\p_x + t)^2 (\p_y + u)$ is 
$\p_{xxxy} + \p_{xxyy} + u \p_{xxx} + (u+2t+s) \p_{xxy} + 2t \p_{xyy}$, there will be a unique choice of $s$, $t$ and $u$ that gives an element of $\cN'$.  
We first take $u$ so that the coefficient of 
$\p_{xxx}$ is zero, then choose $t$ so the 
coefficient of $\p_{xyy}$ is zero, and 
finally choose $s$ so that the coefficient of 
$\p_{xxy}$ is zero.
Thus the hypotheses of Theorem \ref{unique_C_theorem} are still met with $\mathcal{C}'$ and
$\mathcal{N}'$, and the coefficients of the principal symbol of the 
unique $N' \in \mathcal{N}'$ give three more upward invariants, for $a_{20}$, $a_{11}$ and $a_{02}$.

To get upward invariants for $a_{10}$ and $a_{01}$, we need to add another term to $\cC'$.  This term should have three free parameters, which we can choose to make the coefficients of $\p_{xx}$, $\p_{xy}$ and $\p_{yy}$ all zero.
While we could use the term 
$(\p_x + \p_y + f)(\p_x + g)(\p_y + h)$, we will instead show another possibility.
Suppose we first focus on a term that would make the coefficients of 
$\p_{xx}$ and $\p_{xy}$ zero.
Since $( 2,0 )$ and $( 1,1 )$ are both covered by the vector
$( 2,1 )$, we can use a term with 
principal symbol $\p_{xxy}$ and two free parameters, such as 
$(\p_x + f)^2 (\p_y + g)$.
But what will we use for a term that makes the
coefficient of $\p_{yy}$ zero, particularly since we should only use one more free parameter?
One solution is to reuse parameters from higher levels, for instance by adding the term
$(\p_x + h)(\p_y + r)^2$.

So this leads us to taking 
\begin{align*}
\cC'' &= \{ (\p_x + q)^3 (\p_y + r)^2 + (\p_x + \p_y + 
s)(\p_x + t)^2 (\p_y + u) + (\p_x + f)^2 (\p_y + g) \\
&\hskip1cm+ (\p_x + h)(\p_y + r)^2
\colon q,r,s,t,u \in K \}    
\end{align*}
and letting 
$\mathcal{N}''$ be the set of elements of $\mathcal{N}$ with coefficients of
zero for all derivatives of order above $1$.
The reader may verify that $\cC''$ is closed under gauge transformations.  
(This would not have been the case if we had instead added the term 
$(\p_x + h)(\p_y + q)^2$.)

Letting some $L \in \cL$ be given, we have that as before, there is a unique way to choose $q$, $r$, $s$, $t$ and $u$ in $C''$ so that $L-C'' \in \cN'$.
Ignoring the term
$(\p_x + f)^2 (\p_y + g) + (\p_x + h)(\p_y + r)^2$
for the moment, let $N'$ be the unique element of $\cN'$
determined by using $q$, $r$, $s$, $t$ and $u$ as above, and letting $f$, $g$ and $h$ be zero.
Then the coefficient of $\p_{xx}$ in $N'$ consists of $a_{20}$ minus an expression in $q$, $r$, $s$, $t$ and $u$.
For simplicity, call this coefficient $a_{20}'$, and define $a_{11}'$ and $a_{02}'$ similarly.

We now need to change $f$, $g$ and $h$ to non-zero values so that the coefficients of the degree 2 terms are zero.
Since $(\p_x + f)^2 (\p_y + g) + (\p_x + h)(\p_y + r)^2$ is 
$\p_xxy + \p_xyy + g \p_{xx} +
(2f+2r) \p_{xy} + h \p_{yy}$,
taking $f$, $g$ and $h$ non-zero gives us that the coefficients of $\p_{xx}$, 
$\p_{xy}$ and $\p_{yy}$ are 
$a_{20}' - g$, $a_{11}' - (2f + 2r)$ and $a_{02}' - h$, respectively.
There is a unique way to make these coefficients zero, we take $g = a_{20}'$,
$f = (a_{11}' - 2r)/2$ and
$h = a_{02}'$.
This uniquely determines the coefficients of $\p_x$ and $\p_y$ in our element of $\cN''$, giving upward invariants for $a_{10}$ and $a_{01}$.

This process would continue, until  $N''' \in \mathcal{N}'''$ is just a function, giving an upward invariant for the constant term.
\end{example}

\begin{example}
Consider the class that is the downward closure of $\{ \p_{xxy}, \p_{xyy} \}$. This class of operators was considered in ~\cite{invariants_gen}, where \emph{obstacles to factorizations}~\cite{obstacle2} were used to compute four Laplace invariants. A fifth invariant  was then found to complete the set but it was not a Laplace invariant. This fifth invariant had a long complicated formula and no meaning or structure. Our new method allows us to construct a complete set of Laplace invariants. 

So $n = 2$ and let $\mathcal{L}$ 
be the class of operators maximally generated by $T$, where
\begin{equation*}
    T =\{ \p_{xxy}, \p_{xyy}  \}
\end{equation*}
%So $\cL$ is the class of operators of the form
%$$\p_{xxy} + \p_{xyy} + a_{20} \p_{xx} + a_{11} \p_{xy} + a_{02}\p_{yy} + a_{10}\p_x + a_{01}\p_y + a_{00} \, .$$
Let 
$L = \p_{xxy} + \p_{xyy} + a_{20} \p_{xx} + a_{11} \p_{xy} + a_{02}\p_{yy} + a_{10}\p_x + a_{01}\p_y + a_{00} \in \mathcal{L}$. 
We have two maximal invariants, but they are both $1$.
There are three submaximal terms, $a_{20} \p_{xx}$, $a_{11} \p_{xy}$ and $a_{02}\p_{yy}$.  
So there will be $3-n = 3-2$ extra invariants, and one compatibility invariant.
We have that when $L = \p_{xxy} + \p_{xyy} + a_{20} \p_{xx} + a_{11} \p_{xy} + a_{02}\p_{yy} + \dots$,
that $e^{-g} L e^g = L' = \p_{xxy} + \p_{xyy} + a'_{20} \p_{xx} + a'_{11} \p_{xy} + a'_{02}\p_{yy} + \dots$.
Here, $a'_{20} = a_{20} + g_y$, $a'_{11} = a_{11} + 2g_y + 2g_x$ and $a'_{02} = a_{02} + g_x$.
These are three equations in the two unknowns $g_x$ and $g_y$.
We get $a'_{20} - a_{20} = g_y$ and $a'_{02} - a_{02} = g_x$, which we substitute into $a'_{11} - a_{11} = 2g_y + 2g_x$ to get  
$a'_{11} - a_{11} = 2(a'_{20} - a_{20}) + 2(a'_{02} - a_{02})$.
Rearranging this, 
$a'_{11} - 2a'_{20} - 2a'_{02}  = a_{11} - 2a_{20} - 2a_{02}$.
This gives the extra invariant, 
$$I_e=a_{11} - 2a_{20} - 2a_{02} \ . $$

Returning to $a'_{20} - a_{20} = g_y$ and $a'_{02} - a_{02} = g_x$, we differentiate both and get 
$a'_{20x} - a_{20x} = g_{xy} = a'_{02y} - a_{02y}$.
Thus $a'_{20x} - a'_{02y} = a_{20x} - a_{02y}$, showing that
$$I_c=a_{20x} - a_{02y}$$ 
is a compatibility invariant.

To get upward invariants, we first take $\cC$ to be the class of operators of the form
$(\p_x + p)(\p_y + q)(\p_x + \p_y + r)$.
Expanding, these operators have the form 

\begin{equation} \label{expansion}
\begin{split}
&\p_{xxy} + \p_{xyy} + q\p_{xx} + (p+q+r)\p_{xy} + p\p_{yy} +(q_x + r_y + q(p+r))\p_x +\\ 
&(q_x + r_x + p(q+r))\p_y + ((pq+q_x)r + q r_x + p r_y + r_{xy})
\end{split}
\end{equation}

Taking $\cN = \{b_{10}\p_x + b_{01}\p_y + b_{00} \colon b_{10},b_{01},b_{00} \in K \}$,
we get $q = a_{20}$, $p = a_{02}$ and $p+q+r = a_{11}$.
The last equation becomes 
$r = a_{11} - p - q = a_{11} - a_{20} - a_{02}$.  
The coefficients $b_{10}$ and $b_{01}$ are then invariants.
We get $$I_{10} = a_{10} -
(q(p+r) + q_x +r_y)$$ and 
$$I_{01} = a_{01} -
(p(q+r) + q_x + r_x) \ .$$

Substituting in our values of $p$, $q$ and $r$, these invariants agree with those obtained given in Theorem 4 of \cite{invariants_gen}.

In particular, we have the associated representation
\begin{equation} \label{ex6:incomplete_fact}
    L = (\p_x + p)(\p_y + q)(\p_x + \p_y + r) + I_{10} \p_x + I_{01} \p_y + b_{00}
\end{equation}
for some $p,q,r,b_{00} \in K$.

To obtain an upward invariant involving $a_{00}$, we let 
$\cC'$ be the class of operators of the form
$(\p_x + p)(\p_y + q)(\p_x + \p_y + r) + (\p_x + s)(\p_y + t)$ 
and take $\cN'$ to be $K$.
The first group of equations is almost the same as before;  we have
$q = a_{20}$, $p = a_{02}$ and $p+q+r+1 = a_{11}$.
Thus $p$ and $q$ have the same values as before while the value of $r$ is 
$r' = a_{11} - a_{20} - a_{02} - 1$, one less than the previous value.
We take our expansion (\ref{expansion}), substitute $r'$ for $r$ in it and add
$(\p_x + s)(\p_y + t) = \p_{xy} + t \p_x + s \p_y + st + t_x$,
giving that the following must be in $\cN'$.
\begin{equation} \label{expansion 2}
\begin{split}
&(a_{10} -(q_x + r'_y + q(p+r') + t))\p_x +
(a_{01}-(q_x + r'_x + p(q+r') + s))\p_y\\ &+ a_{00}-(pq+q_x)r' + q r'_x + p r'_y + r'_{xy} + st + t_x)
\end{split}
\end{equation}

Thus $t = a_{10} -(q_x + r'_y + q(p+r')) = a_{10} -(q_x + r_y + q(p+(r-1))) = I_{10} + q$.
Similarly, $s = a_{01}-(q_x + r'_x + p(q+r') = I_{01} + p$.
The terms without derivation operators make our desired invariant, it is 
\begin{equation*}
%%\begin{split}
%&
I_{00} = a_{00}-((pq+q_x)r' + q r'_x + p r'_y + r'_{xy} + st + t_x) \ .
%=\\
%&a_{00}-((pq+q_x)(r-1) + q r_x + p r_y + r_{xy} + (I_{01} + p)(I_{10} + q) + (I_{10}+q)_x)
%\end{split}
\end{equation*}
In particular, for this invariant we  have the associated representation
\begin{equation*}
    L = (\p_x + p)(\p_y + q)(\p_x + \p_y + r) + (\p_x + s)(\p_y + t)+ I_{00} 
\end{equation*}
for some $p,q,r,s,t \in K$, where again
$p,q,r$ here can be different from $p,q,r$ in~\eqref{ex6:incomplete_fact}.
\end{example}

Here is another interesting example from the literature;  its maximal terms do not all have the same degree.
A complete set of invariants for it was obtained in \cite{inv_cond_non_hyper_case}.
%, but these invariants also appear in Theorem 2 of \cite{shemya:invariant properties}.

\begin{example}
\label{X3example}

So $n = 2$ and 
let $\mathcal{L}$ 
be the class of operators maximally generated by $T$, where
\begin{equation*}
    T = \{ \p_{xxx}, a_{11}\p_{xy}, a_{02}\p_{yy} \} \ .
\end{equation*}
%Let $\cL$ be the downward closure of $\{\p_{xxx}, a_{11}\p_{xy}, a_{02}\p_{yy} \}$, 
So elements of $\cL$ have the form
\begin{equation}
\label{L for X^3}
\p_{xxx} + a_{20}\p_{xx} + a_{11}\p_{xy} + a_{02}\p_{yy} + a_{10}\p_x + a_{01}\p_y + a_{00}.
\end{equation}

We have that $1$, $a_{11}$, $a_{02}$ 
are maximal invariants.
Applying a gauge transformation to~\eqref{L for X^3}, 
%$L = \p_{xxx} + a_{20}\p_{xx} + a_{11}\p_{xy} + a_{02}\p_{yy} + a_{10}\p_x + %a_{01}\p_y + a_{00}$, 
we get 
$L' = e^{-g} L e^g = \p_{xxx} + a'_{20}\p_{xx} + a_{11}\p_{xy} + a_{02}\p_{yy} + a'_{10}\p_x + a'_{01}\p_y + a'_{00}$,
where $a'_{20} = a_{20} + 3g_x$, $a'_{10} = a_{10} + 3(g_x^2 + g_{xx}) + 2 a_{20}g_x + a_{11}g_y$ and
$a'_{01} = a_{01} + a_{11}g_x + 2a_{02}g_y$.
%We leave for now the middle equation as it has the term $3(g_x^2 + g_{xx})$ that is not linear in $g_x$. 
%From the other two,
%we get $a'_{20} - a_{20} = 3g_x$ and 
%$a'_{01} - a_{01} = a_{11}g_x + 2a_{02}g_y$.
%Note that $a_{11}$ and $a_{02}$ are invariants.
Solving for $g_x$ and $g_y$, we get $(a'_{20} - a_{20})/3 = g_x$ and 
$(a'_{01} - a_{01} - a_{11}(a'_{20} - a_{20})/3)/(2a_{02}) = g_y$.
Compatibility condition ($g_{xy}=g_{yx}$) then gives us 
%These give $a'_{20}_y/3 - a_{20}_y/3 = g_{xy} = 
%(a'_{01}/(2a_{02}))_x - (a_{01}/(2a_{02}))_x -
%(a_{11}(a'_{20} - a_{20})/(6a_{02}))_x$.
%Clearing fractions and rearranging,
$2a'_{20y} - 3(a'_{01}/a_{02})_x + (a_{11}a'_{20}/a_{02})_x =  
2a_{20y} - 3(a_{01}/a_{02})_x +
(a_{11}a_{20}/a_{02})_x$.
This implies that 
$$
I_{10} = 2a_{20y} - 3(a_{01}/a_{02})_x +
(a_{11}a_{20}/a_{02})_x
$$
is a compatibility invariant.

To get an upward invariant for $a_{10}$, we need at least $n = 2$ free parameters in $\cC$,
which will zero out two submaximal terms in going from $\cL$ to $\cN$.
Accordingly, we take the classes of operators 
%of the form
$$
\cC = \{ 
(\p_x + p)^3 + a_{11} (\p_x + p)(\p_y + q) + a_{02} (\p_y + q)^2  \}\ ,
\ 
\cN = \{ b_{10}\p_x + b_{00} \} \ .
$$
%and let $\cN$ be the class of operators of the form
%$b_{10}\p_x + b_{00}$.

%Now $(\p_x + p)^3 + a_{11} (\p_x + p)(\p_y + q) + a_{02} (\p_y + q)^2$ expands to
%\begin{equation}
%\label{expansion for X^3}
%\begin{split}
%&(\p_x^3 + 3p \p_x^2 + 3(p_x + p^2) \p_x + p^3 + 3pp_x + p_{xx}) +\\ 
%&a_{11}(\p_{xy} + q\p_x + p\p_y + pq + q_x) +
%a_{02}(\p_y^2 + 2q \p_y + q^2 + q_y).
%\end{split}
%\end{equation}

Subtracting an element of $\cC$ 
%Subtracting this 
from (\ref{L for X^3}), we get 
\begin{equation}
\label{X^3 difference}
\begin{split}
&(a_{20} - 3p) \p_{xx} +
(a_{10} - 3(p_x + p^2) - a_{11}q) \p_{x} +
(a_{01} - a_{11}p - 2a_{02}q) \p_y +\\
&a_{00} - (p^3 + 3pp_x + p_{xx} + a_{11}(pq + q_x) + a_{02}(q^2 + q_x)).
\end{split}
\end{equation}

For this to be in $\cN$, we need $a_{20} = 3p$ and
$a_{01} - a_{11}p = 2a_{02}q$.
These give $p = a_{20}/3$ and
$q = (a_{01} - a_{11}a_{20}/3)/(2a_{02})$.
The coefficient of $\p_x$ in (\ref{X^3 difference}) then becomes the invariant
\begin{equation*}
\begin{split}
I_{10} &= a_{10} - 3(p_x + p^2) - a_{11}q\\
&= a_{10} - 3((a_{20}/3)_x + (a_{20}/3)^2) - a_{11}(a_{01} - a_{11}a_{20}/3)/(2a_{02})
\end{split}
\end{equation*}

Similarly, we get the invariant $I_{01}$ by taking 
$\cN$ to be operators of the form $b_{01}\p_y + b_{00}$, and making 
(\ref{X^3 difference}) an element of this $\cN$.
This gives us $p = a_{20}/3$ as before, and 
$q = (a_{10} - 3(p_x + p^2)) / a_{11}$.
The coefficient of $\p_y$ then becomes the invariant
\begin{equation*}
\begin{split}
I_{01} &= a_{01} - a_{11}p - 2a_{02}q\\
&= a_{01} - a_{11}a_{20}/3 - 2a_{02}(a_{10} - 3((a_{20}/3)_x + (a_{20}/3)^2)) / a_{11} \ .
\end{split}
\end{equation*}
In particular, we have associated representation
\begin{equation} \label{ex7:incomplete_fact}
    L = (\p_x + p)^3 + a_{11} (\p_x + p)(\p_y + q) + a_{02} (\p_y + q)^2  + I_{10} \p_x + I_{01} \p_y + b_{00}
\end{equation}
for some $p,q,r,b_{00} \in K$.

These are not strictly speaking upward invariants, since $I_{10}$ also involves $a_{01}$ and $I_{01}$ also involves $a_{10}$.
However, they can be manipulated to yield two upward invariants.
We have 
\begin{equation*}
\begin{split}
I_{10} &= (a_{10} - a_{11}a_{01}) - (3((a_{20}/3)_x + (a_{20}/3)^2) - a_{11}a_{11}a_{20}/(6a_{02})) 
\mbox{\quad and}\\
I_{01} &= (a_{01} - 2a_{02}a_{10})
- (a_{11}a_{20}/3 - 6a_{02}((a_{20}/3)_x + (a_{20}/3)^2)) / a_{11})
\end{split}
\end{equation*}
Thus $(I_{10} + a_{11}I_{01})/(1-2a_{11}a_{02})$ is an upward invariant for $a_{10}$, and
$(I_{01} + 2a_{02}I_{10})/(1-2a_{11}a_{02})$ is an upward invariant for $a_{01}$.
(We leave the case where $(1-2a_{11}a_{02}) = 0$ to the reader.)

To get an upward invariant for $a_{00}$, we slightly modify $\cC$ by changing the $a_{11}(\p_x + p)(\p_y + q)$ term to $a_{11}(\p_x + {r})(\p_y + q)$.
This gives that $\cC'$ is the class of operators of the form 
$(\p_x + p)^3 + a_{11} (\p_x + { r})(\p_y + q) + a_{02} (\p_y + q)^2$. This changes the difference in (\ref{X^3 difference}) to become
\begin{equation}
\label{X^3 difference 2}
\begin{split}
&(a_{20} - 3p) \p_{xx} +
(a_{10} - 3(p_x + p^2) - a_{11}q) \p_{x} +
(a_{01} - a_{11}{r} - 2a_{02}q) \p_y +\\
&a_{00} - (p^3 + 3pp_x + p_{xx} + a_{11}({r}q + q_x) + a_{02}(q^2 + q_x)).
\end{split}
\end{equation}
Now we take $\cN$ to be the set of terms of the form $b_{00}$.
Then $p = a_{20}/3$ as before, and
$q = (a_{10} - 3(p_x + p^2)) / a_{11}$ as in the derivation of $I_{01}$.
Setting the coefficient of $\p_y$ equal to zero, we get 
$r = (a_{01}-2a_{02}q)/a_{11}$.
Then 
$$I_{00} = a_{00} - (p^3 + 3pp_x + p_{xx} + a_{11}(rq + q_x) + a_{02}(q^2 + q_x)).$$

In particular, we  have associated representation
\begin{equation*}
    L = (\p_x + p)^3 + a_{11} (\p_x + { r})(\p_y + q) + a_{02} (\p_y + q)^2+ I_{00} 
\end{equation*}
for some $p,q,r,s,t \in K$, where again
$p,q,r$ here can be different from $p,q,r$ in~\eqref{ex6:incomplete_fact}.

%This can be given a more explicit form by substituting in the above values of $p$, $q$ and $r$.
\end{example}

The next example is a simple version of one treated by Mironov in \cite{mironov2009invariants} and by Athorne and Yilmaz in \cite{AthorneYilmaz2016}.
Consider the operator that is the downward closure of $\{\p_{x_1x_2 \dots x_n} \}$ in dimension $n$.
Mironov gets invariants for the case $n = 4$, while Athorne and Yilmaz produce invariants for all cases through $n = 6$.
%Their methods work better than ours does in this highly symmetrical situation.

\begin{example}
\label{3d hyperbolic example}
Let $n=3$, and call independent variables $x$, $y$ and $z$.
Suppose $\mathcal{L}$ 
is the class of operators maximally generated by $T$, where
\begin{equation*}
    T = \{ \p_{xyz} \} \ ,
\end{equation*}
%We let $\cL$ be the downward closure of $\{\p_{xyz}\}$, 
so elements of $\cL$ have the form
$\p_{xyz} + a_{110}\p_{xy} + a_{101}\p_{xz} +a_{011}\p_{yz} +
a_{100}\p_{x} +
a_{010}\p_{y} +
a_{001}\p_{z} +
a_{000}$.

Since $n = 3$ and there are $3$ submaximal terms, we have $1$ as a maximal invariant, no extra invariants, and three compatible invariants.
\begin{equation*}
\begin{split}
I_{cxy} &= a_{011y} - a_{101x}\\
I_{cxz} &= a_{011z} - a_{110x}\\ %\mbox{ \quad and}\\ 
I_{cyz} &= a_{101z} - a_{110y}\\
\end{split}
\end{equation*}

We take $\cC$ to be the set of operators of the form
$(\p_x + p)(\p_y + q)(\p_z + r)$,
and get that $L - C$ is
\begin{equation}
\label{3d hyperbolic L-C}
\begin{split}
&(a_{110}-r)\p_{xy} + 
(a_{101}-q)\p_{xz} +
(a_{011}-p)\p_{yz} +
(a_{100}-qr-r_y)\p_{x} +\\
&(a_{010}-pr-r_x)\p_{y} +
(a_{001}-pq-q_x)\p_{z} +\\
&a_{000}-(pqr+ r q_x + p r_y + q r_x + r_{xy})
\end{split}
\end{equation}

Taking $\cN$ to be the set of operators of the form
$b_{100}\p_x + b_{010}\p_y + b_{001}\p_z + b_{000}$,
we get 
$$ p = a_{011} \hspace{2cm} 
q = a_{101} \hspace{2cm} %\mbox{and}
r = a_{110}.
$$
The coefficients of $\p_x$, $\p_y$ and $\p_z$ now give us the three upward invariants
\begin{equation*}
\begin{split}
&I_{100} = a_{100}-a_{101} a_{110}-a_{110y}\\
&I_{010} = a_{010}- a_{011} a_{110}-a_{110x}\\
&I_{001} = a_{001}-a_{011} a_{101}-a_{101x}
\end{split}
\end{equation*}
These are essentially the same as those in the literature, and they correspond to the representation
\begin{equation*}
    L = I_{100}\p_x + I_{010}\p_y + I_{001}\p_z + b_{000}
\end{equation*}
for some $b_{000} \in K$.

To get $I_{000}$, there are several possibilities for a class of expressions to add to those in $\cC$.
\begin{equation*}
\begin{split}
&(\p_x + s)(\p_y + t) + 
(\p_x + s)(\p_z + u) +
(\p_y + t)(\p_z + u) =\\
&\p_{xy} + \p_{xz} + \p_{yz} +
(t+u)\p_x + (s+u)\p_y +
(s+t)\p_z +\\ 
&(st + t_x + su + u_x + tu + u_y)
\end{split}    
\end{equation*}
would work, but substituting $s$, $t$ and $u$ into $st + t_x + su + u_x + tu + u_y$ could produce complicated expressions.
So we use the following, which avoids products such as $st$.
\begin{equation}
\label{hyperbolic second term}
\begin{split}
&(\p_x + s)(\p_y + q) + 
(\p_y + t)(\p_z + r) +
(\p_z + u)(\p_x + p) =\\
&\p_{xy} + \p_{xz} + \p_{yz} +
(q+u)\p_x + (r+s)\p_y +
(p+t)\p_z +\\ 
&(sq + q_x + tr + r_y + pu + p_z)
\end{split}    
\end{equation}
Note that the order of the factors in the three terms is chosen to produce 
$q_x$, $r_y$ and $p_z$.  The other order would produce $s_x$, $t_y$ and $u_z$, which would yield more complicated expressions.

So we let $\cC'$ be the set of operators of the form
$$
(\p_x + p)(\p_y + q)(\p_z + r) +
(\p_x + s)(\p_y + q) + 
(\p_y + t)(\p_z + r) +
(\p_z + u)(\p_x + p)
$$
This gives us that $L-C$ is
\begin{equation}
\label{hyperbolic second L-C}
\begin{split}
&(a_{110}-r-1)\p_{xy} + 
(a_{101}-q-1)\p_{xz} +
(a_{011}-p-1)\p_{yz} +\\
&(a_{100}-qr-r_y-q-u)\p_{x} +
(a_{010}-pr-r_x-r-s)\p_{y} +\\
&(a_{001}-pq-q_x-p-t)\p_{z} +\\
&a_{000}-(pqr+ r q_x + p r_y + q r_x + r_{xy} + sq + q_x + tr + r_y + pu + p_z)
\end{split}    
\end{equation}
Letting $\cN'$ be the set of operators of the form $b_{000}$, we have 
$$ p = a_{011} - 1 \hspace{2cm} 
q = a_{101} - 1\hspace{2cm} %\mbox{and}
r = a_{110} - 1.
$$
Next we get 
\begin{equation*}
%\label{hyperbolic s, t and u}
\begin{split}
&s = a_{010}-pr-r_x-r\\
&t = a_{001}-pq-q_x-p\\
&u = a_{100}-qr -r_y-q
\end{split}    
\end{equation*}

Substituting these all into the constant term of (\ref{hyperbolic second L-C}), we get 
\begin{equation*}
\begin{split}
I_{000} = a_{000}-(&pqr+ r q_x + p r_y + q r_x + r_{xy} + sq + q_x + tr + r_y + pu + p_z) \ ,
%= a_{= a_{000}-(&(a_{011} - 1)(a_{101} - 1)(a_{110} - 1)+ (a_{110} - 1) a_{101}_x +\\ 
%
%000}-(&(a_{011} - 1)(a_{101} - 1)(a_{110} - 1)+ (a_{110} - 1) a_{101}_x +\\ 
%&(a_{011} - 1) a_{110}_y + (a_{101} - 1) a_{110}_x + a_{110}_{xy} + \\
%&(a_{010}-(a_{011} - 1)(a_{110} - 1)-a_{110}_x - (a_{110} - 1))(a_{101} - 1) + \\
%&a_{101}_x + (a_{001}-(a_{011} - 1)(a_{101} - 1) - a_{101}_x - \\
%&(a_{011} - 1))(a_{110} - 1) + a_{110}_y + \\
%&(a_{011} - 1)(a_{100}-(a_{101} - 1)(a_{110} - 1) - a_{110}_y - (a_{101} - 1)) + a_{011}_z)
\end{split}    
\end{equation*}
which is associated to the representation
\begin{equation*}
    L = (\p_x + p)(\p_y + q)(\p_z + r) +
(\p_x + s)(\p_y + q) + 
(\p_y + t)(\p_z + r) +
(\p_z + u)(\p_x + p) + I_{000} 
\end{equation*}
for some $p,q,r,s,t,u \in K$.

Unlike Athorne and  Yilmaz's
technique, our method does not naturally produce invariants that are symmetric in all variables.
Adding $I_{cxz} - (1/3)(I_{cxz})_y -
(1/3)(I_{cyz})_x - I_{100} - I_{010}
-I_{001} - 1$ and simplifying, we get
$$
a_{000}-(a_{100}a_{011}+ a_{010}a_{101} + a_{001}a_{110} 
- 2a_{011}a_{101}a_{110} +
(a_{110xy} + a_{101xz}
+ a_{011yz})/3)
$$
This is essentially the same as the corresponding invariant in~\cite{Athorne2018}.
\end{example}
\begin{example}
\label{inductive hyperbolic example}
As an interesting application, we can also obtain an inductive definition of upward invariants for the bottom terms of any totally hyperbolic operator as in (\ref{3d hyperbolic example}).
We let $\cL_n$ be the downward closure of $\p_{x_1x_2 \dots x_n}$ for some $n \ge 2$.
As noted in~\cite{Athorne2018}, the form of an upward invariant for a given term $t$ only depends on how far it is below the maximal element of $\cL_n$.
This is because an upward invariant for $t$ only depends on the coefficient of $t$ and terms above it, so an upward invariant for any term the same distance below $\p_{x_1x_2 \dots x_n}$ as $t$ is can be obtained by substituting the corresponding coefficients in an upward invariant for $t$.
For example, when $n=2$, an upward invariant for $a_{00}$ is 
$$I_{00} = a_{00} - (a_{10}a_{01} + a_{10x_1}),$$
and for any $n \ge 2$, upward invariants for terms that are two levels below $\p_{x_1x_2 \dots x_n}$ can be obtained by substitution in it.
For example, when $n=3$ we have the invariant
$$I_{001} = a_{001} - (a_{101}a_{011} + a_{101 x_1}).$$

This means that once we have upward invariants for the bottom terms for every $n \ge 2$, we can easily construct a complete set of invariants for any $n$.
So suppose we have some $n \ge 2$ and an upward invariant for the bottom term of $\cL_n$,
$$I_{00\dots0} = a_{00\dots0} - E,$$
where $E$ is an expression in the other coefficients of $\cL_n$.
For clarity, let us denote the coefficients in $\cL_n$ using $b$'s, while still using $a$'s for the coefficients in $\cL_{n+1}$.

Now we seek an upward invariant for the bottom term in $\cL_{n+1}$, and accordingly let $\cC$ be the class of operators of the form
$$(\p_z + p) \cL_n
= (\p_z + p)(\p_{x_1x_2\dots x_n} + b_{011\dots1} \p_{x_2x_3\dots x_n} \dots  + b_{00\dots0}),$$
where we use $z$ to denote $x_{n+1}$.
We will take $\cN$ to be the class of all operators in $\cL_{n+1}$ that have $a_{11\dots 110} = 1$ and 
all terms involving $\p_z$ equal to zero.  
Note that $\cN$ is the same as $\cL_n$, except that its coefficients have different names.
In particular, we have the invariant 
\begin{equation}
\label{invariant for N}
I_{00\dots0} = b_{00\dots0} - E
\end{equation}
where $E$ is an expression in the $b_\alpha$ for $\alpha$ an $n$-long string of $0$'s and $1$'s that is not all $0$'s.

Expanding $\cC$, we have that it is
\begin{equation}
\label{C for recursion}
\begin{split}
&\p_{x_1x_2 \dots x_n z} +
 b_{011\dots 1}\p_{x_2x_3 \dots x_n z} + \dots
b_{00\dots 0} \p_z +\\
& p \p_{x_1x_2 \dots x_n} +
(p b_{011\dots 1} + b_{011\dots 1z})\p_{x_2x_3 \dots x_n} + \dots
(p b_{00\dots 0} + b_{00\dots 0z})\\
\end{split}
\end{equation}
Thus to make $L-C$ be in $\cN$, we take $b_{\alpha} = a_{\alpha 1}$ for each $n$-long string of $0$'s and $1$'s $\alpha$ other than $11 \dots 11$, and we also take
$p = a_{11 \dots 110} - 1$.

This gives us that $N = L-C$ is
\begin{equation}
\label{N for recursion}
\begin{split}
& \p_{x_1x_2 \dots x_n} +
(a_{011\dots 110} - p b_{011\dots 1} - b_{011\dots 1z})\p_{x_2x_3 \dots x_n} + \dots
(a_{00 \dots 00} - p b_{00\dots 0} - b_{00\dots 0z})\\
\end{split}
\end{equation}
Substituting the coefficients of (\ref{N for recursion}) into (\ref{invariant for N}) gives the desired upward invariant for the bottom term of $\cL_{n+1}$.

For example, when $n = 2$ we have 
\begin{equation}
\label{level 2 invariant}
I_{00} = a_{00} - (a_{10}a_{01} + a_{10 x_1}).
\end{equation}
To get an invariant for $n=3$ from this, we have $p = a_{110} - 1$, $b_{ij} = a_{ij1}$, and 
$a_{ij} = a_{ij0} - p b_{ij} - b_{ijz} = 
a_{ij0} - (a_{110}-1) a_{ij1} - a_{ij1z}$.
Substituting these into (\ref{level 2 invariant}), we get
\begin{equation}
\begin{split}
I_{000} = &(a_{000} - ( (a_{110}-1) a_{001} + a_{001z}) -\\
&((a_{100}- (a_{110}-1) a_{101} - a_{101z})
(a_{010}- (a_{110}-1) a_{011} - a_{011z})+\\
&(a_{100}- (a_{110}-1) a_{101} - a_{101z})_{x_1}),\\
\end{split}    
\end{equation}
where we are using $z$ for $x_3$.

While this is not a symmetrical expression in the coefficients, this recursive definition could be of %theoretical 
interest.
\end{example}

\section{Complete sets of Laplace invariants: a constructive proof}
\label{sec:proof} 

We will be working toward a proof that when
$\cL$ is maximally generated and approximately flat, that there is a complete set of invariants for $\cL$.  
Our construction of invariants will start with the highest degree terms in $\cL$, and work down.
We will of course include the maximal invariants in our complete set of invariants $I$.
Next, we have the following.

\begin{definition}\label{Delta_definition}
Let $\cL$ be maximally generated, let $L$ and $L'$ be two arbitrary elements of $\cL$.
Assume $L'$ is a gauge transform of $L$, so $L' = e^{-g} L e^g$ for some $g \in K$.  
Let $E$ be some 
expression in coefficients of $L$ (which may involve algebraic operations and differentiation), and let $E'$ be the same expression in the corresponding coefficients in $L'$.
Then the {\em difference} of $E$, $\Delta E$, is given by $\Delta E = E' - E$.
\end{definition}

\begin{theorem}\label{Delta_properties_theorem}
Let $E$ and $F$ be expressions as above, and let $a_{\vecc{v}} \p_{\vecc{v}}$ be a term of $\cL$.
Then the following hold.
\begin{enumerate}
    \item $\Delta (E+F) = \Delta E + \Delta F$.
    \item $\Delta (E_{x_i}) = (\Delta E)_{x_i}$ for any variable $x_i$.
    \item $E$ is invariant iff $\Delta E = 0$ for all $L$ and $L'$.
\end{enumerate}
\end{theorem}
The proof is straightforward.
%\begin{proof}
%Straightforward computations.
%are 
%left to the reader.
%\end{proof}

%\cat{Everything after here has been changed.  The problem was that nothing works unless $\cL$ is framed, so that concept has to be introduced earlier.}

In Examples \ref{submaximal_example} through \ref{X3example}, we produced compatability invariants by solving a system of equations for the derivatives $g_{x_i}$.  
We need to carefully examine when this can be done.
With assumptions as in Definition \ref{Delta_definition}, make the additional assumption that $\cL$ is approximately flat.  
%Then every submaximal term is only covered by maximal terms.  
Let $s$ be the number of submaximal terms in $\cL$, and let $\vecc{v}$ be one of the $s$ submaximal vectors.
For this $\vecc{v}$, let $T$ be the set of $i$ so that $\vecc{v}+e_i$ is the vector of a maximal term.
In this situation, an easy calculation shows that
\begin{equation} \label{a_sub_v_equation}
\Delta a_{\vecc{v}} = \sum_{i \in T}\; (\vecc{v}(i) +1)
a_{\vecc{v}+\vecc{e}_i}\; g_{x_i} =
(\sum_{i \in T}\; (\vecc{v}(i) +1) a_{\vecc{v}+\vecc{e}_i} \;\vecc{e}_i)\cdot \nabla g = \phi(\vecc{v}) \cdot \nabla g
\end{equation}
Here, $\nabla g = (g_{x_1}, g_{x_2}, \dots g_{x_n})$ is the vector of partial derivatives of $g$, and 
$\phi(\vecc{v})$ denotes the vector $(\sum_{i \in T}\; (\vecc{v}(i) +1) a_{\vecc{v}+\vecc{e}_i} \;\vecc{e}_i)$
for this submaximal vector $\vecc{v}$ and its set $T$.

There will be $s$ equations of this form, one for each submaximal vector $\vecc{v}$.
We need a condition on the set of vectors $\phi(\vecc{v})$ so that some $n$ of the vectors in the set yield equations \eqref{a_sub_v_equation} that determine $\nabla g$.
This condition is obviously that the entire set of the $\phi(\vecc{v})$ spans $K^n$.  

\begin{definition} \label{framed_definition}
For a maximally generated class $\cL$, let $S$ be the set of its submaximal vectors.  For each $\vecc{v} \in S$, let 
$\phi(\vecc{v})$ be the vector $(\sum_{i \in T}\; (\vecc{v}(i) +1) a_{\vecc{v}+\vecc{e}_i}\; \vecc{e}_i)$ as above.  We say that $\cL$ is {\em framed} iff $\{ \phi(\vecc{v}) \colon \vecc{v} \in S \}$ spans $K^n$. 
\end{definition}

%as in (4) of Theorem \ref{Delta_properties_theorem}, although $\vecc{v}$ and $S$ will vary between equations.
%Note that all of the vectors $\vecc{v}+\vecc{e}_i$ are maximal, so the corresponding coefficients $a_{\vecc{v}+\vecc{e}_i}$ are all invariants.
%Since $\cL$ is downward closed, each maximal term has $n$ or fewer submaximal terms below it.
%\cat{Since $\cL$ is approximately flat, $s \geq n$ and each of the $n$ partial derivatives $g_{x_i}$ of $g$ appears in at least one of the $s$ equations.}

%These $s$ equations are linear in the $n$ partial derivatives $g_{x_i}$.
%Solving any $n$ of them usually gives each $g_{x_i}$ equal to a linear expression in various $\Delta a_{\vecc{v}}$ with invariant coefficients.
%\cat{There are cases where the coefficients of the maximal terms are such that the linear system does not completely determine the values of the $g_{x_i}$. We need to investigate these further.  (Everything from here on has been rewritten a bit.)}

Assuming $\cL$ is framed, we have a set of $n$ equations of the form
$\Delta a_{\vecc{v}} = \sum_{i \in S}\; (\vecc{v}(i) +1)
a_{\vecc{v}+\vecc{e}_i}\; g_{x_i}$
that determine all of the derivatives 
$g_{x_i}$ in terms of $n$ of the 
$\Delta a_{\vecc{v}}$.
In addition to these $n$ equations, we have $s - n$ ``extra'' equations which give other of the $\Delta a_{\vecc{v}}$ as linear expressions in the $g_{x_i}$ with invariant coefficients.

All of the above equations yield invariants.
To simplify notation, we illustrate this by letting $n = 3$, calling the three variables $x$, $y$ and $z$, letting $a$, $b$ and $c$ be coefficients of submaximal terms, and letting $\alpha$, $\beta$, $\gamma$ and $\delta$ be invariant coefficients.
Then from the $n$ equations that look like 
$g_x = \alpha \Delta a + \beta \Delta b$, 
$g_y = \gamma \Delta b + \delta \Delta c$, and so on, we construct invariants as follows.
Differentiating, and setting compatibility partials equal, we get equations like
$(\alpha \Delta a + \beta \Delta b)_y = g_{xy} = (\gamma \Delta b + \delta \Delta c)_x$.
This becomes
$\alpha_y \Delta a + \alpha \Delta a_y + \beta_y \Delta b + \beta \Delta b_y =
\gamma_x \Delta b + \gamma \Delta b_x + \delta_x \Delta c + \delta \Delta c_x$.
Which is 
$\alpha_y (a'-a) + \alpha (a_y' - a_y) + \beta_y (b'- b) + \beta (b_y'- b_y) =
\gamma_x (b'- b) + \gamma (b_x'- b_x) + \delta_x (c'- c) + \delta (c_x'- c_x)$.
Rearranging, we get 
$\alpha_y a' + \alpha a_y' + \beta_y b' + \beta b_y' - 
(\gamma_x b' + \gamma b_x' + \delta_x c' + \delta c_x')=
\alpha_y a + \alpha a_y + \beta_y b + \beta b_y - 
(\gamma_x b + \gamma b_x + \delta_x c + \delta c_x)$, showing that 
$(\alpha_y a + \alpha a_y + \beta_y b + \beta b_y) - 
(\gamma_x b + \gamma b_x + \delta_x c + \delta c_x)$
is an invariant.
In general there are $n (n-1)/2$ compatibility invariants that look like this, one for each pair of variables.

For the $s-n$ ``extra'' equations which look like
$\alpha \Delta a + \beta \Delta b = \gamma \Delta b + \delta \Delta c$, we proceed as follows.  
We expand the expressions with $\Delta$, and get 
$\alpha (a'-a) + \beta (b'- b) = \gamma (b'- b) + \delta (c'-c)$.
Rearranging,
$\alpha a' + \beta b' - \gamma b' - \delta c' =
\alpha a + \beta b - \gamma b - \delta c $, which makes 
$\alpha a + \beta b - \gamma b - \delta c$ invariant.

The above discussion gives us the following.
\begin{theorem}\label{enough_top_invariants_theorem}
If $\cL$ is maximally generated, approximately flat, and framed, it has all of the maximal, extra and compatibility invariants needed to produce a set of invariants that is complete by Theorem \ref{sufficient_for_a_complete_set_theorem}.
\end{theorem}

Our next step is to produce a set of upward invariants for all the terms that are not maximal or submaximal.  
%We can do this whether or not $\cL$ is approximately flat, but can not prove the resulting set of invariants is complete without this assumption.
Our method will be to repeatedly apply Theorem \ref{unique_C_theorem}.
The first step is to construct a class $\cC$ so that the unique $N$ produced as in Theorem \ref{unique_C_theorem} is in the class $\cN$ of operators in $\cL$ that have all their coefficients of maximal and submaximal terms equal to zero.  
Our construction will have a distinguished set of submaximal terms, which need a certain property so that we can use them to build a ``framework''.

\begin{definition}\label{framing_set_definition}
Let $\cL$ be a maximally generated class of operators that is approximately flat and framed.
Then the set $S$ of vectors of submaximal terms of $\cL$ is such that
$\{ \phi(\vecc{v}) \colon \vecc{v} \in S \}$ spans $K^n$.
This is equivalent to there being an $n$-element subset  
$\{ \vecc{v}_1, \vecc{v}_2, \dots \vecc{v}_n \}$ of $S$ where 
the set $\{ \phi(\vecc{v}_1), \phi(\vecc{v}_2), \dots \phi(\vecc{v}_n) \}$
is linearly independent.
We call such a set $\{ \vecc{v}_1, \vecc{v}_2, \dots \vecc{v}_n \}$ of submaximal vectors a {\em framing set} for $\cL$.
\end{definition}

The vast majority of operators in the literature give maximally generated classes that have framing sets and are thus framed.

\begin{theorem} \label{single_maximal_implies_framed_theorem}
Let $\cL$ be the class of operators that is generated by a single nonzero term $a_{\vecc{v}} \p^{\vecc{v}}$, where $\vecc{v}(i) > 0$ for all $i \leq n$.
Then $\cL$ is framed. 
\end{theorem}
\begin{proof}
A framing set consists of 
the $n$ vectors for submaximal terms
$\{ \vecc{v} - \vecc{e}_i \colon i \leq n \}$, since we have that each
$\phi(\vecc{v} - \vecc{e}_i)$ is a nonzero multiple of $\vecc{e}_i$.
\end{proof}

\begin{example} \label{not_framed_example}
Here is an example of a class $\cL$ that is not framed.
We take $n = 2$, write $x$ for $x_1$ and $y$ for $x_2$.  
Then we let $\cL$ be  maximally generated by 
$\{ \p_{xx}, 2 \p_{xy}, \p_{yy} \}$, so operators in $\cL$ have the form
$\p_{xx} + 2 \p_{xy} + \p_{yy} + a_{10} \p_{x} + a_{01} \p_{y} + a_{00}$.
There are two submaximal vectors, $( 1,0 )$ and $( 0,1 )$, and 
$\phi(( 1,0 )) = \phi(( 0,1 )) = ( 2,2 )$.

%This is also an example 
%is an obstacle to the method we are using to prove sets of invariants complete, rather than 
%of a class of operators without a complete set of (algebro-differential) invariants.
%In Theorem \ref{sufficient_for_a_complete_set_theorem}, we proved completeness by recovering the gauge function $g$ from the values of the invariants.
%The problem seems to be that 
%Here $g$ is not completely determined in this case, since when $L \in \cL$
%has $a_{10} = a_{01} = 0$, we have that 
%$g$ is determined only up to knowing $g_x +g_y$.
%I'll show the calculation for this, which also proves that the class has no complete set of invariants.
\end{example}

\begin{theorem}\label{zero_max_and_submax_theorem}
Let $\cL$ be maximally generated, approximately flat, and framed.
Let $\cN$ be the class of $L \in \cL$ where all the coefficients of maximal and submaximal terms of $L$ are zero.
Then there is a class $\cC$ of operators so that for every $L \in \cL$ there is a unique $C \in \cC$ so that $N = L-C$ is in $\cN$.
\end{theorem}

\begin{proof}
%The notation is simplified if we work with vectors, where the vector $\vecc{v}$ corresponds to the term $a_{\vecc{v}} \p^{\vecc{v}}$.
%We let $\vecc{e}_i$ denote the vector that is all zeroes, except that its $i$-th component is $1$.
Let $M$ be the set of maximal vectors for $\cL$, and let $S$ be the set of submaximal vectors.
Then every vector in $S$ is covered by at least one vector in $M$, and every vector in $M$ covers at least one vector in $S$.
(If a vector such as $k \vecc{e}_1 = ( k,0,0, \dots 0 )$ is maximal, it only covers the one submaximal vector $(k-1) \vecc{e}_1$.)

We will first produce a correspondence between elements of $M$ and subsets of $S$ that has the properties needed to construct expressions in $\cC$.
We may assume that the set of vectors for $\cL$ contains nonzero multiples of all the $\vecc{e}_i$, since we may simply ignore variables whose derivative symbols do not appear in $\cL$.
%Thus for each $i \leq n$ there is a maximal vector $\vecc{m}_i$ that is above $\vecc{e}_i$.
%Let $k_i$ be $\vecc{m}_i (i)$.
%We have that the vector $\vecc{v}_i = \vecc{m}_i - \vecc{e}_i$ is submaximal.
%Let $K$ be $\{ \vecc{v}_i \colon i \leq n \}$.
%This set may of course have repeated elements.
Since $\cL$ is framed, it has a framing set
$\{ \vecc{v}_1, \vecc{v}_2, \dots \vecc{v}_n \}$. 
We will use this to define a set of $n$ distinguished parameters,
$\{ c_1,c_2, \dots c_n \}$.

Let $S'$ be $S - \{ \vecc{v}_1, \vecc{v}_2, \dots \vecc{v}_n \}$.
Now fix some function $f \colon S' \rightarrow M$ %where $f(\vecc{m}_i - \vecc{e}_i) = \vecc{m}_i$ for all $i$ and $f$ 
which takes every submaximal vector in $S'$ to a maximal vector that covers it.
To each $\vecc{m} \in M$, we associate the set $f^{-1}(\vecc{m})$, the preimage of $\vecc{m}$.
Some sets $f^{-1}(\vecc{m})$ may be empty, but the ones that are not partition $S'$.

We will construct the class $\cC$ as a set of sums of operators, where there will be one operator for each vector in $M$.
For each $\vecc{m} \in M$, the corresponding operator will have principal symbol 
$a_{\vecc{m}} \p^{\vecc{m}}$.  
This guarantees that for $L \in \cL$ and 
$C \in \cC$, all operators of the form 
$N = L - C$ will have coefficients of zero in all terms corresponding to maximal vectors.

For each of the submaximal vectors $\vecc{v}_i$,  we will make sure that the operator for each
maximal vector $\vecc{m} = \vecc{v}_i + \vecc{e}_j$
that covers it has a factor of 
$(\p_{x_j} + c_j)^{\vecc{m}(j)}$.
Looking at some particular $\vecc{v}_i$, only the operators for maximal vectors $\vecc{m}$ that cover $\vecc{v}_i$ will contribute terms in $\cC$
corresponding to the vector $\vecc{v}_i$.
When $\vecc{m} = \vecc{v}_i + \vecc{e}_j$, the term 
contributed by the operator for $\vecc{m}$ will be
$\vecc{m}(j) a_{\vecc{m}} c_j \p^{\vecc{v}_i} = 
(\vecc{v}(i)+1) a_{\vecc{v}+\vecc{e}_i} c_j \p^{\vecc{v_i}}$.
To make the term corresponding to $\vecc{v}_i$ zero in $L - C$, we must have 
$\sum_{j \in T(\vecc{v}_i)} (\vecc{v}(j)+1) a_{\vecc{v}+\vecc{e}_j} c_j = a_{\vecc{v}_i}$,
where $T(\vecc{v}_i)$ is the set of $j$ so that 
$\vecc{v}_i + \vecc{e}_j$ covers $\vecc{v}_i$.
Letting $\vecc{c} = ( c_1,c_2, \dots c_n )$,
this is the equation 
$\phi(\vecc{v}_i) \cdot \vecc{c} = a_{\vecc{v}_i}$.

Then to make all the coefficients of all the $\vecc{v}_i$ terms zero in $L-C$, 
we have $\phi(\vecc{v}_i) \cdot \vecc{c} = a_{\vecc{v}_i}$
for all $i \leq n$.
Since $\{ \vecc{v}_1, \vecc{v}_2, \dots \vecc{v}_n \}$ is a framing set, 
the $n$ vectors $\phi(\vecc{v}_i)$ are linearly independent, and this system has a unique solution for the $c_i$.
The values of the $c_i$ will be fixed by this, so we may henceforth treat them as constants.

Let $M$ be the set of maximal vectors of $\cL$.
We will define an operator $F_\vecc{m}$ for each $\vecc{m} \in M$, and then define $\cC$ to be the class of sums of the form  $\sum_{\vecc{m} \in M} F_\vecc{m}$. 
For any given $\vecc{m} \in M$, $F_\vecc{m}$ will be obtained by slightly modifying 
$a_{\vecc{m}} \prod_{i \leq n} ( \p_{x_i} + c_i )^{\vecc{m}(i)}$ .
Observe that the principal symbol of this expression is already $a_{\vecc{m}} \p^\vecc{m}$, as desired.
%so putting an operator like this in 
%$\cC = \sum_{\vecc{m} \in M} F_\vecc{m}$ will have the effect of naturally removing all the terms $a_{\vecc{m}} \p^\vecc{m}$
%from $\cN = \cL - \cC$.

Our goal is to choose $\cC$ so that all of the submaximal terms of $\cL$ will also be removed.
The parameters $c_i$ let us remove the $n$ submaximal terms with vectors $v_1, v_2, \dots v_n$.
To remove the terms with vectors in $S' = S - \{ v_1,v_2,\dots v_n \}$, we introduce a parameter $p_\vecc{v}$ for each $v \in S'$.
This parameter will be placed into the operator $F_\vecc{m}$, where $\vecc{m} = f(\vecc{v})$.
Specifically, we will let $j$ be such that $\vecc{m} = \vecc{v} + \vecc{e}_j$, take $a_{\vecc{m}} \prod_{i \leq n} ( \p_{x_i} + c_i )^{\vecc{m}(i)}$, and replace one factor of 
$( \p_{x_j} + c_j )$ in it by $( \p_{x_j} + p_\vecc{v} )$.
In practice, this will be slightly more complicated because a given $F_\vecc{m}$ may have factors $( \p_{x_j} + c_j )$ replaced for several $j$ at once.

Consider any maximal vector $\vecc{m} \in M$. 
Let $S(\vecc{m})$ be the set of indices where  we will replace factors in $a_{\vecc{m}} \prod_{i \leq n} ( \p_{x_i} + c_i )^{\vecc{m}(i)}$,
so $S(\vecc{m}) = \{ i \colon \vecc{m} - \vecc{e}_i \in f^{-1}(\vecc{m}) \}$.
Let $\vecc{u}_{\vecc{m}}$ be the vector $\vecc{m} - \sum_{i \in S(\vecc{m})} \vecc{e}_i$. 
Finally, let the operator $F_{\vecc{m}}$ be
$a_{\vecc{m}} \prod( \p_{x_i} + c_i )^{\vecc{u}_{\vecc{m}}(i)} \: 
\prod_{i \in S(\vecc{m})} (\p_{x_i} + p_{\vecc{m}-\vecc{e}_i})$.
It has principal symbol $a_{\vecc{m}} \p^{\vecc{m}}$, as desired.  
Now we let $\cC$ be given by
\begin{equation} \label{C_equation}
\cC = \sum_{\vecc{m} \in M} F_{\vecc{m}} = 
\sum_{\vecc{m} \in M} a_{\vecc{m}} \prod( \p_{x_i} + c_i )^{\vecc{u}_{\vecc{m}}(i)} 
\prod_{i \in S(\vecc{m})} (\p_{x_i} + p_{\vecc{m}-\vecc{e}_i})
\end{equation}

Given a submaximal vector $\vecc{v}$, we have two cases.
First assume that $\vecc{v} = \vecc{v}_i$ for some $i$.
Then $\vecc{v}$ is not in the domain of $f$, and for any 
maximal vector $\vecc{m}$ covering $\vecc{v}$ where 
$\vecc{m} = \vecc{v} + \vecc{e}_i$, we have $i \notin S(\vecc{m})$.
So for each of these $\vecc{m}$ we have $\vecc{u}_{\vecc{m}}(i) = \vecc{m}(i)$, and
$F_\vecc{m} = a_
\vecc{m} (\p_{x_i} + c_i )^{{\vecc{m}}(i)} G$, where $\p_{x_i}$ does not appear in $G$ and $G$ has principal symbol $\p^{\vecc{v}-\vecc{v}(i) \vecc{e}_i}$.
Then the $\vecc{v}$ term in this $F_\vecc{m}$ is 
$a_\vecc{m} \vecc{m}(i) c_i \p^\vecc{v} = (\vecc{v}(i)+1) a_{\vecc{v}+\vecc{e}_i} c_j \p^{\vecc{v_i}}$, as desired, and summing these over all 
maximal vectors $\vecc{m} \succ \vecc{v} = \vecc{v}_i$ gives us the equation $\phi(\vecc{v}_j) \cdot \vecc{c} = a_{\vecc{v}_j}$, which is true by our choice of the $c_j$.

For the second case, assume that $\vecc{v}$ is not one of the $\vecc{v}_i$.
If $\vecc{m} = \vecc{v} + \vecc{e}_j = f(\vecc{v})$,
then $j \in S(\vecc{m})$ and $F_\vecc{m}$ is
$a_\vecc{m} (\p_{x_j} + c_j)^{\vecc{v}(j)} (p_{x_j} + p_\vecc{v}) G$, where
$\p_{x_j}$ does not appear in $G$ and $G$ has principal symbol $\p^{\vecc{v}-\vecc{v}(j) \vecc{e}_j}$.
So this $F_\vecc{m}$ has $\vecc{v}$ term
$a_\vecc{m} (\vecc{v}(j) c_j + p_\vecc{v})$.
If $\vecc{m} \neq f(\vecc{v})$, then the calculation in the previous paragraph gives us that the $\vecc{v}$ term in this $F_\vecc{m}$ is 
$a_\vecc{m} \vecc{m}(i) c_i \p^\vecc{v}$, a term that does not contain $p_\vecc{v}$.
Summing over all maximal vectors $\vecc{m}$ above $\vecc{v}$, we get that the coefficient of the $\vecc{v}$ term in $\cC$ is 
$a_\vecc{m} (\vecc{v}(j) c_j + p_\vecc{v})$ plus an expression in which $p_\vecc{v}$ does not appear.
There is a unique choice of $p_\vecc{v}$ that makes this equal to $a_\vecc{v}$.

\end{proof}

%Let us call vectors that are covered by submaximal vectors {\em subsubmaximal}.
The above theorem gives us that all of the coefficients of maximal terms in the unique $N \in \cN$ are invariants.
By the construction in the proof, each of these coefficients depends only on coefficients of maximal and submaximal terms of $\cL$, and is thus an upward invariant.
%\begin{definition}
%In a maximally generated class of operators, we define the {\em level} of a vector recursively as follows.  Maximal vectors have level $0$, and the level of all other vectors is $1$ less than the minimum of the levels of the vectors that cover them.
%\end{definition}
%In other words, the level of a vector is the length of the longest upward path from it to a maximal vector.
%Submaximal vectors have level less than or equal to $-1$.
%If the class of operators is approximately flat, the levels of all submaximal vectors are actually $-1$. 

%In the above theorem, all of the vectors for nonmaximal terms in $\cN$ have level $-3$  or less,
%since they are covered by maximals of $\cN$, which are covered by submaximals, which are covered by maximals.
\begin{theorem} \label{main_theorem}
Let $\cL$ be maximally generated, framed and approximately flat.
Then $\cL$ has a complete set of invariants.
\end{theorem}
\begin{proof}
By Theorem \ref{enough_top_invariants_theorem}, $\cL$ has enough maximal, extra and compatibility invariants.  
We may choose particular extra and compatibility invariants for definiteness.
It remains to produce upward invariants for all terms of $\cL$ that are not maximal or submaximal.

%The proof of Theorem \ref{zero_max_and_submax_theorem} gives us that the maximal terms of $\cN$ are invariants.
%These are upward invariants for the corresponding terms in $\cL$.  

%We claim that every term of $\cL$ that is not maximal or submaximal has an upward invariant, and will construct them.
Using the fact that $\cL$ is framed and the construction in Theorem \ref{zero_max_and_submax_theorem} there is a class $\cC_{m}$, containing operators 
which uniquely determine the parameters $\{ c_1, c_2, \dots c_n \}$, where each $c_i$ only appears in factors of $(\p_{x_i} + c_i)$.
Furthermore, the class $\cC_m$ is such that for all $L \in \cL$ there is a unique element of $\cC_m$ so that subtracting it from $L$ makes all of its maximal and submaximal terms zero.

Let $V$ be the set of vectors that are not maximal or submaximal.
Fix a function $f$ that takes each vector in $V$ to a vector that covers it, noting that there are no maximal vectors in the range of $f$.
For each $\vecc{v} \in V$, define the operator $B_\vecc{v}$ to be 
$(\p_{x_j} + q_\vecc{v}) \prod_{i \leq n} (\p_{x_i} + c_i)^{\vecc{v}(i)}$, 
where $j$ is such that $f(\vecc{v}) = \vecc{v} + \vecc{e}_j$ and $q_\vecc{v}$ is a free parameter.

The idea is that $B_\vecc{v}$ is designed so that having it as a summand of $\cC$ makes it so that there is a unique choice of $q_\vecc{v}$ that makes the $\vecc{v}$ term in $\cN = \cL - \cC$ zero.
Observe that the $\vecc{v}$ term of $B_\vecc{v}$ is $q_\vecc{v} \p^\vecc{v}$, and that the other terms containing $q_\vecc{v}$ are all below the $\vecc{v}$ term.
Note also that things are complicated by the fact that the principal symbol of each $B_\vecc{v}$ is $\p^\vecc{w}$ for some $\vecc{w}$ that covers $\vecc{v}$ in $\prec$.
To remove the $\vecc{v}$ term from $\cL$, we have to add a term $\p^\vecc{w}$ where $\vecc{w} \succ \vecc{v}$.

Since $\cL$ is maximally generated, we have that it is closed under adding operators of the form $B_\vecc{v}$, since these are not defined for maximal vectors $\vecc{v}$.

Now let a vector $\vecc{v}$ be given that is not maximal or submaximal in $\cL$.
We construct an upward invariant for it as follows.
Let $\cC$ be the class of all operators produced by adding an operator of the form $\sum_{\vecc{w} \succ \vecc{v}} B_\vecc{w}$ to an operator in $\cC_m$.
Let $\cN$ consist of the operators in $\cL$ that have all of their maximal and submaximal terms zero, and where also all of their terms with vectors above $\vecc{v}$ are zero.
Then $\vecc{v}$ is a maximal vector of $\cN$, and will be an upward invariant for the $\vecc{v} $ term, provided that there is a unique way to chose the parameters in $C \in \cC$ so that $L-C$ is in $\cN$.

To demonstrate that this choice is unique, we will show how to structure things so that the value of each parameter is forced.
Let $W$ be the set of vectors $\vecc{w}$ where we are adding the operator $B_\vecc{w}$ to $\cC_m$, so $W = \{ \vecc{w} \in V \colon \vecc{w} \succ \vecc{v} \}$.
Let $L \in \cL$ be given, and let $L' = L - \sum_{\vecc{w} \in W} \p^{f(\vecc{w})}$.
Then $L'$ is in $L$.

As in the proof of Theorem \ref{zero_max_and_submax_theorem}, we have unique choices of the $c_i$ and  $p_\vecc{u}$ appearing in $\cC_m$ so that subtracting this element of $\cC_m$ from $L'$ makes all of its maximal and submaximal terms zero.
Fix these choices, giving a particular element $C_m$ of $\cC_m$, so we have that $L' - C_m = L - (C_m + \sum_{\vecc{w} \in W} \p^{f(\vecc{w})}) $ 
has values of zero for all the terms corresponding to those that are maximal or submaximal in $\cL$.
Since $\sum_{\vecc{w} \in W} \p^{f(\vecc{w})}$ agrees with $\sum_{\vecc{w} \succ \vecc{v}} B_\vecc{w}$ on maximal and submaximal terms of $\cL$, we also have that 
$L - (C_m + \sum_{\vecc{w} \succ \vecc{v}} B_\vecc{w})$ is zero on these terms.

We will now pick the values of the parameters $q_\vecc{w}$ in the $B_\vecc{w}$, being careful about the order in which we do this.
Extend the partial order $\succ$ to a total order $>$ on vectors in $W$, and consider elements of $W$ in decreasing order relative to $>$.
At each step, we will choose the value of $q_\vecc{w}$ to make the $\vecc{w}$ term in $L - (C_m + \sum_{\vecc{w} \succ \vecc{v}} B_\vecc{w})$ equal to zero.

Looking at any $B_\vecc{w} = (\p_{x_j} + q_\vecc{w}) \prod_{i \leq n} (\p_{x_i} + c_i)^{\vecc{w}(i)}$,
we see that it contains only the one undetermined function $q_\vecc{w}$, since all of the $c_i$ have already been fixed.
We also have that $q_\vecc{w}$ only appears in terms with vectors $\vecc{w}$ and below.

Let $\vecc{w}$ be any element of $W$, and consider the step in the process where we choose $q_\vecc{w}$.
The $\vecc{w}$ term of $L - (C_m + \sum_{\vecc{w} \succ \vecc{v}} B_\vecc{w})$
may contain the $c_i$ and $p_\vecc{u}$, but those are already fixed.
It may also contain $q_{\vecc{w}'}$ for ${\vecc{w}'} \succ \vecc{w}$, but then ${\vecc{w}'} > \vecc{w}$, and $q_{\vecc{w}'}$ has already been determined.
Since the coefficient of the $\vecc{w}$ term consists of $q_\vecc{w}$ added to an expression that does not contain $q_\vecc{w}$, there is a unique value of $q_\vecc{w} $ that makes this term zero.
So our construction works, and produces an upward invariant for every term that is not maximal or submaximal.

Adding these invariants to the others gives a set of invariants that is complete by Theorem \ref{sufficient_for_a_complete_set_theorem}.
\end{proof}

\section*{Acknowledgement}
This material is based upon work supported by the National Science Foundation under
grant No.1708033.

%%%%%%%%%%%%%%%%%%%%%%%%%%%%%%%
%\section{The form of invariants}
%All known invariants in the literature are polynomials in coefficients of the operator and their derivatives, which we will call {\em polynomial invariants}.
%While it seems obvious that all invariants should essentially be of this form, this has not been proved.
%In Definition \ref{invariant definition}, we defined an  invariant as an ``expression'' in coefficients and their derivatives.  
%Making this more precise, an {\em invariant} is a differentiable function of coefficients of the operator and their derivatives.
%It is trivially true that any differentiable function of a set of invariants is itself an invariant.
%Our goal is to prove that every invariant is a differentiable function of a set of polynomial invariants. 

% 

% \bibliographystyle{plain}
% \bibliography{/home/ekaterina/Desktop/work/general.bib

\begin{thebibliography}{10}

\bibitem{Athorne2018}
Ch. Athorne.
\newblock {Laplace maps and constraints for a class of third-order partial
  differential operators}.
\newblock {\em Journal of Physics A: Mathematical and Theoretical}, 51(8),
  2018.

\bibitem{AthorneYilmaz2016}
Ch. Athorne and H.~Yilmaz.
\newblock Invariants of hyperbolic partial differential operators.
\newblock {\em Journal of Physics A: Mathematical and Theoretical},
  49(13):135201, 2016.

\bibitem{AthorneYilmaz2019}
Ch. Athorne and H.~Yilmaz.
\newblock Twisted {{L}}aplace maps.
\newblock {\em Journal of Physics A: Mathematical and Theoretical},
  52(22):225201, 2019.

\bibitem{Darboux2}
G.~Darboux.
\newblock {\em {Le{\c c}ons sur la th{\'e}orie g{\'e}n{\'e}rale des surfaces et
  les applications g{\'e}om{\'e}triques du calcul infinit{\'e}simal}},
  volume~2.
\newblock Gauthier-Villars, 1889.

\bibitem{dzhokhadze2004invariants}
O.~Dzhokhadze.
\newblock Laplace invariants for some classes of linear partial differential
  equations.
\newblock {\em Differential Equations}, 40(1):63--74, 2004.

\bibitem{FO1}
M.~Fels and P.~J. Olver.
\newblock {Moving coframes. {I}. {A} practical algorithm}.
\newblock {\em Acta Appl. Math.}, 51(2):161--213, 1998.

\bibitem{FO2}
M.~Fels and P.~J. Olver.
\newblock {Moving coframes. {I}{I}. {R}egularization and theoretical
  foundations}.
\newblock {\em Acta Appl. Math}, 55:127--208, 1999.

\bibitem{ganzha2013intertwining}
Elena~I. Ganzha.
\newblock Intertwining laplace transformations of linear partial differential
  equations.
\newblock In Moulay Barkatou, Thomas Cluzeau, Georg Regensburger, and Markus
  Rosenkranz, editors, {\em Algebraic and Algorithmic Aspects of Differential
  and Integral Operators}, pages 96--115. Springer Berlin Heidelberg, 2014.

\bibitem{2015:super}
S.~Hill, E.~Shemyakova, and Th. Voronov.
\newblock {Darboux transformations for differential operators on the
  superline}.
\newblock {\em Russian Mathematical Surveys}, 70(6):207--208, 2015.
% \newblock \texttt{arXiv:1505.05194 [math.MP]}.

\bibitem{shemya:hobby2016:iterated}
D.~Hobby and E.~Shemyakova.
\newblock {Classification of multidimensional Darboux transformations: first
  order and continued type}.
\newblock {\em SIGMA (Symmetry, Integrability and Geometry: Methods and
  Applications)}, 13(10):20 pages, 2017.
% \newblock \texttt{arXiv:1605.04362 [math.DG]}.

\bibitem{HOVHANNISYAN20161690}
G.~Hovhannisyan and O.~Ruff.
\newblock Darboux transformations on a space scale.
\newblock {\em Journal of Mathematical Analysis and Applications},
  (2):1690--1718, 2016.

\bibitem{HOVHANNISYAN2018776}
G.~Hovhannisyan, O.~Ruff, and Z.~Zhang.
\newblock Higher dimensional {{D}}arboux transformations.
\newblock {\em Journal of Mathematical Analysis and Applications},
  (1):776--805, 2018.

\bibitem{VanHoeij_transformations2017}
E.~Imamoglu and M.~van Hoeij.
\newblock Computing hypergeometric solutions of second order linear
  differential equations using quotients of formal solutions and integral
  bases.
\newblock {\em Journal of Symbolic Computation}, 83:254--271, 2017.

\bibitem{shemya:voronov2016:berezinians}
S.~Li, E.~Shemyakova, and Th. Voronov.
\newblock {Differential operators on the superline, Berezinians, and Darboux
  transformations}.
\newblock {\em Lett. Math. Phys.}, 107(9):1689--1714, 2017.
% \newblock \texttt{arXiv:1605.07286 [math.DG]}.

\bibitem{Mansf_book}
E.~L. Mansfield.
\newblock {\em {A Practical Guide to the Invariant Calculus}}.
\newblock Cambridge University Press, 2010.

\bibitem{mironov2009invariants}
A.~Mironov.
\newblock On the {{L}}aplace invariants of a fourth-order equation.
\newblock {\em Differential Equations}, 45(8), 2009.

\bibitem{OP:09}
P.~J. Olver and J.~Pohjanpelto.
\newblock {Differential invariant algebras of {L}ie pseudo-groups}.
\newblock {\em Adv. Math.}, 222(5):1746--1792, 2009.

\bibitem{inv_cond_non_hyper_case}
E.~Shemyakova.
\newblock {Invariant properties of third-order non-hyperbolic linear partial
  differential operators}.
\newblock In {\em {Lecture Notes in Computer Science}}, volume 5625, pages
  154--169. Springer Berlin Heidelberg, 2009.

\bibitem{2013:invertible:darboux}
E.~Shemyakova.
\newblock {Invertible {D}arboux transformations}.
\newblock {\em SIGMA (Symmetry, Integrability and Geometry: Methods and
  Applications)}, 9:Paper 002, 10, 2013.

\bibitem{shem:darboux2}
E.~Shemyakova.
\newblock {Proof of the completeness of {D}arboux {W}ronskian formulae for
  order two}.
\newblock {\em Canad. J. Math.}, 65(3):655--674, 2013.

\bibitem{shemyakova2013_DT_fact}
E.~Shemyakova.
\newblock Classification of {{D}}arboux transformations for operators of the
  form $\partial_{x}\partial_{y}+a\partial_{x}+b\partial_{y}+c$.
\newblock {\em Illinois J. Math.}, 64(1):71--92, 04 2020.
% \newblock \texttt{arXiv:1304.7063 [math.MP].}

\bibitem{movingframes}
E.~Shemyakova and E.~L. Mansfield.
\newblock {Moving frames for {L}aplace invariants}.
\newblock In {\em {Proceedings of the twenty-first international symposium on
  symbolic and algebraic computation} (ISSAC)}, pp. 295--302, New York, 2008. ACM.

\bibitem{ShemyakovaVoronov_Sturm_Liou_densities2018}
E.~Shemyakova and Th. Voronov.
\newblock Differential operators on the algebra of densities and factorization
  of the generalized {{S}}turm--{{L}}iouville operator.
\newblock {\em Lett. Math. Phys.}, 109(2):403--421, 2019.
% \newblock \texttt{arXiv:1710.09542 [math.DG]}.

\bibitem{invariants_gen}
E.~Shemyakova and F.~Winkler.
\newblock {A Full System of Invariants for Third-Order Linear Partial
  Differential Operators in General Form}.
\newblock {\em Lecture Notes in Comput. Sci.}, 4770:360--369, 2007.

\bibitem{obstacle2}
E.~Shemyakova and F.~Winkler.
\newblock {Obstacles to the factorization of linear partial differential
  operators into several factors}.
\newblock {\em Programming and Computer Software}, 33(2):67--73, 2007.
% \newblock \texttt{http://arxiv.org/abs/1010.2652 [math.AP]}.

\bibitem{inv_cond_hyper_case}
E.~Shemyakova and F.~Winkler.
\newblock {On the Invariant Properties of Hyperbolic Bivariate Third-Order
  Linear Partial Differential Operators}.
\newblock vol. 5081 of {\em {Lecture
  Notes in Computer Science}}, pp. 199--212. Springer, 2007.

\bibitem{Smirnov2018}
S.~V. Smirnov.
\newblock Factorization of {{D}}arboux---{{L}}aplace transformations for
  discrete hyperbolic operators.
\newblock {\em Theor. Math. Phys.}, 199(2):621--636, 2019.

\end{thebibliography}
% }

\end{document}